\documentclass[reqno, 12 pt]{amsart}%amsart

\usepackage{amsthm}
\usepackage{amssymb}
\usepackage{latexsym}
\usepackage{url}
\usepackage{amsmath}
\usepackage{verbatim}
\usepackage{graphicx}
\usepackage{epsf}
\usepackage{enumerate}
\usepackage{geometry}
\usepackage[width=7cm]{caption}
\usepackage{hyperref}
\hypersetup{breaklinks=true,
pagecolor=white,
colorlinks=true,%false
citecolor=black,%
filecolor=black,%
linkcolor=black,%
urlcolor=black
}
\newtheorem{thm}{Theorem}[section]
\newtheorem{cor}[thm]{Corollary}
\newtheorem{lem}[thm]{Lemma}
\newtheorem{prop}[thm]{Proposition}

\newtheorem{defin}[thm]{Definition}

\newtheorem{condition}[thm]{Condition}

\theoremstyle{definition}
\newtheorem{example}[thm]{Example}

\newcommand{\dom}{\textnormal{dom}}

\newcommand{\Int}{\textnormal{int}}

\newcommand{\R}{\mathbb{R}}

\newcommand{\N}{\mathbb{N}}

\newcommand{\wt}{\widetilde}
\newcommand{\wh}{\widehat}

\title[The geometric stability of Voronoi diagrams]{The geometric stability of Voronoi diagrams in normed spaces which are not uniformly convex}
\author{Daniel Reem}
\thanks{
\noindent IMPA - Instituto Nacional de Matem\'atica Pura e Aplicada, Estrada Dona Castorina 110, Jardim Bot\^anico, CEP 22460-320, Rio de Janeiro, RJ,  Brazil.\\
\noindent E-mail: dream@impa.br }


\date{April 29, 2013}
\begin{document}
\maketitle
\begin{abstract}
The Voronoi diagram is a geometric object which is widely used in many areas. 
 Recently it has been shown that under mild conditions Voronoi diagrams have a certain continuity property:  
small perturbations of the sites yield small perturbations in the shapes of the corresponding Voronoi cells. 
However, this result is based on the assumption that the ambient normed space is uniformly 
convex. Unfortunately, simple counterexamples show that if uniform convexity is removed, then instability 
 can occur. Since Voronoi diagrams in normed spaces which are not uniformly  convex  
 do appear in theory and practice, e.g., in the  plane with the Manhattan ($\ell_1$) distance, it is natural to ask whether the stability property can be generalized  to them, perhaps under additional assumptions. This paper shows that this is indeed the  case  assuming  the unit sphere of the space has a certain (non-exotic) structure and the sites satisfy a certain 
 ``general position'' condition related to it. The condition on the unit sphere is that 
 it can be decomposed into at most one ``rotund part'' and at most finitely many non-degenerate convex parts. 
 Along the way certain topological  properties of Votonoi cells (e.g., that the induced bisectors are 
 not ``fat'') are proved. 
\end{abstract}
\newpage

%%
%%
%%More precisely, the unit sphere should 
 %be decomposed into a (possibly empty) ``rotund part'' and finitely many (perhaps 0) ``flat parts'', 
 %and  no two points from different sites can form a nondegenerate line segment parallel to a line 
 %segment contained in the unit sphere. 
%%
%%
%%

\section{Introduction}\label{sec:Intro}
\subsection{Background} 
Given a world (space) $X$, together with a distance function and a collection of sites (subsets) $(P_k)_{k\in K}$ in 
the world, the Voronoi cell $R_k$ associated with the site $P_k$ is the set 
of all the points in $X$ whose distance to $P_k$ is not greater than their distance to the other sites $P_j, j\neq k$.
 Voronoi diagrams have many theoretical and practical applications in various fields 
  \cite{Aurenhammer,AurenhammerKlein,CSKM2013,ConwaySloane,VoronoiCVD_Review,VoronoiWeb,GruberLek,OBSC} and therefore have been 
  investigated in an extensive way. 
  However, most of this investigation has been focused on finite dimensional Euclidean spaces with point sites, and in many  cases only in $\R^2$ and $\R^3$. Research works studying these diagrams in spaces which are  not Euclidean do exist, e.g.  \cite{Aronov2002,BCMS,BSTY1998,ChewDrysdale,CKSTW1998,Jeong,Klein,KLN2009, KopeckaReemReich,Le1996, Lee,LevenSharir1987,LPL2011,ReemISVD09,SharirAgarwal,TeleaWijk}, but they constitute a small fraction comparing with the research in the Euclidean case and they mainly focus on algorithmic or combinatorial aspects of these diagrams. 
  As a result, not much is known about Voronoi diagrams beyond Euclidean spaces and beyond the studied aspects. 

This paper studies Voronoi diagrams in a class of finite dimensional normed spaces which are not necessarily Euclidean. 
More specifically, it studies a certain continuity property of the cells, namely that small changes of the  sites, e.g., of their position or shape, yield small changes in  the shapes of the corresponding Voronoi cells. 
Recently \cite{ReemGeometricStabilityArxiv} it has been shown that the Voronoi cells do have such a stability 
property. Moreoever, it actually holds in a quite general  setting (e.g., infinitely many sites of a general form in a possibly infinite  dimensional  space) and the bounds are dimension free (no ``curse of dimensionality'' occurs). 
However, an important assumption needed for formulating this result is 
that the space satisfies a certain property called uniform convexity. Roughly speaking, this means that 
not only the unit sphere does not contain any line segment, but actually that 
given any line segment whose endpoints are on the unit sphere, if the length of this segment 
is bounded below by a certain positive number $\epsilon$, then the middle point of this segment 
 must penetrate  the unit ball by at least a positive number $\delta(\epsilon)$, uniformly for 
 all such line segments. Speaking technically, for each $\epsilon\in (0,2]$ there exists $\delta=\delta(\epsilon)\in (0,1]$  such that for all $x,y$ satisfying $|x|=|y|=1$ and $|x-y|\geq \epsilon$, 
 the inequality $|(x+y)/2|\leq 1-\delta$ holds. While, in particular, the familiar Lebesgue spaces $L_p(\Omega)$ and sequence 
 spaces $\ell_p$, $1<p<\infty$ are uniformly convex \cite{BL2000,Clarkson,GoebelReich,LindenTzafriri,Prus2001}, 
 there are useful and not less familiar spaces which are not uniformly convex. Among them, $\R^m$ ($m\geq 2$) with 
 the $\ell_1$ (Manhattan) and $\ell_{\infty}$ (max) norm.

Voronoi diagrams based on non-uniformly convex norms (mainly $\ell_1$ and $\ell_{\infty}$) 
do occur in theory and practice. A simple example is for giving a rough 
estimate on the number of customers of a given facility (e.g., a shopping center/post office) 
 located, together with its competitors,  in a flat region in which transportation is restricted to 
lines parallel to the standard axes. This is the case of Manhattan. See Figures \ref{fig:Shops2D-L1}-\ref{Shops2D-L1Pertube} for an illustration. A few additional examples can be given, e.g., in motion planning and robotics  \cite{LevenSharir1987,SchwartzSharir1987,SchwartzSharir1990} in which a convex robot is 
supposed to move in an environment with obstacles (the sites) 
and the motion path is based on the Voronoi cells and on a distance function induced by the structure of the robot (see also  Section ~\ref{sec:GeneralPosition}); in computer graphics \cite{Hausner} in which the sites are certain points  in an  image and the (centroidal $\ell_1$) Voronoi cells are tiles having good shapes which are the building blocks for a mosaic; for storage \cite{LeeWong} ($\ell_1,\ell_{\infty}$); in relation with VLSI design \cite{Papadopoulou2011,PapadopoulouLee} in which the sites are non-point components and  the $\ell_{\infty}$ Voronoi cells allow one to estimate critical areas in the circuit; in databases \cite{OnishiHoshi2008}, where the sites are vectors in a vector space and the intersections of the Voronoi cells with  respect to varying $\ell_p$ norm $p=1,2,3,\ldots, \infty$ are used for constructing a certain index structure; and possibly in analyzing (approximate) nearest neighbor search algorithms in high dimensional spaces with the $\ell_1$ or the $\ell_{\infty}$ norms \cite{AltLitan,AMNSW1998,Clarkson1999,DIIM2004,Indyk2001}.

%%%%%%%%%%%%%%%%%%%%%%%%%%%%%%%%%%%%%%%%%%%%%%%%%%%%%%%%%%%%%%%%%%%%
\begin{figure}[t]
\begin{minipage}[t]{0.45\textwidth}
\begin{center}
{\includegraphics[scale=0.5]{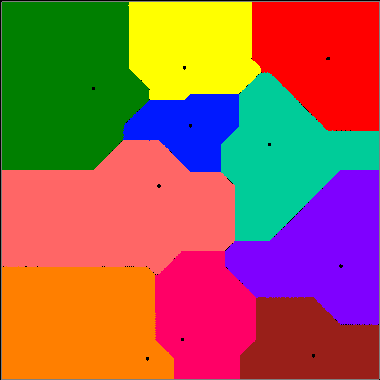}}%{alg40.png}}
\end{center}
 \caption{10 shops/post offices in a flat city, modeled as points. Distance are measured using the $\ell_1$ distance. 
 The Voronoi cells of the sites are displayed.}
\label{fig:Shops2D-L1}
\end{minipage}
\hfill
\begin{minipage}[t]{0.45\textwidth}
\begin{center}
{\includegraphics[scale=0.5]{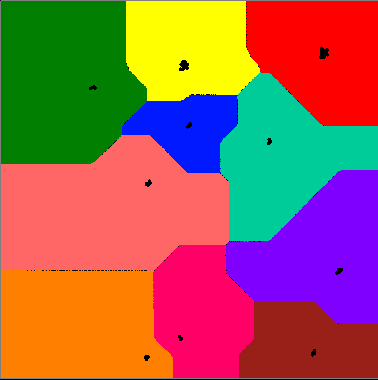}}
\end{center}
 \caption{A more realistic description: the sites have slightly moved (since their location 
 is not known exactly) and their shapes are not an ideal point. But the Voronoi cells have almost the 
 same shapes as before.}
\label{Shops2D-L1Pertube}
\end{minipage}
\end{figure}
%%%%%%%%%%%%%%%%%%%%%%%%%%%%%%%%%%%%%%%%%%%%%%%%%%%%%%%%%%%%%%%%%%%%

In some of the above examples imprecision (noise, input or computational errors, etc.) is an integral 
part of the setting, but  despite this, analysis of the geometric stability of the cells seems to be absent 
here (but see Section \ref{sec:GeneralPosition}) and in most of the literature dealing with Voronoi diagrams 
(with the  exceptions of the very brief and intuitive discussions in \cite{AGMR},\cite[p. 366]{Aurenhammer},\cite{Kaplan} 
in Euclidean settings;  see \cite{ReemGeometricStabilityArxiv}).   
As a matter of fact,  issues related  to the stability/robustness of geometric structures and algorithms due to imprecision are not so common  in  the literature, at least comparing to the majority of works in which ideal conditions (infinite precision of the input, infinite precision in the calculations, etc.) are assumed. 
There are works which do deal in one way or another with such issues, e.g., 
 \cite{AGGKKRS, AGMR, AttaliBoissonnatEdelsbrunner, BandyopadhyaySnoeyink, BLMM2011, CazalDreyfus2010, 
 ChazalCohenSteinerLieutier, SteinerEdelsbrunerHarer, Halperin2002, HalperinShelton1998, HarPeled, HoffmannHopcroftKarasick1988, 
 KhanbanEdalat2003, LofflerPhD, LofflerKreveld,MLvK2010, GuibasSalesinStolfi, Smith2009, SugiharaIriInagakiImai,VGT,Weller}, 
but the  focus is mainly on issues related to combinatorial properties (in Euclidean spaces). 

To the best of the author knowledge, even when geometric stability is discussed in settings closely 
related to Voronoi diagrams, e.g., in \cite{ChazalSouflet2004} (medial axis) and focuses on 
the Hausdorff distance, as done here (see Subsection \ref{subsec:Contribution} below), the question of stability of the shapes of the Voronoi  cells 
under small perturbations of the sites have not been addressed or even raised at all. In addition, the results 
in such papers do not imply the results described in this paper even for the Euclidean case (for instance, in \cite{ChazalSouflet2004}  
one cannot take the set $\Omega$ considered there to be the union of the sites or the complement 
of this union, because $\Omega$ should be a bounded open set with a boundary which is a smooth enough 
manifold, while usually neither the union of the sites nor its complement are such a set).

Taking into account the above, one may be interested in having a theoretical result ensuring the geometric 
stability of the Voronoi cells in normed spaces which are not uniformly convex. At first it seems that 
establishing such a result is impossible, since there are simple counterexamples even in the plane 
with the $\ell_{\infty}$ norm showing that instability of the Voronoi cells may occur under arbitrary 
small perturbations of the sites \cite{ReemGeometricStabilityArxiv} (see also Example \ref{ex:2D-L1-20Cell} below). 
 However, a careful inspection of these counterexamples 
shows that the sites form certain degenerate configurations. Therefore, it is natural to ask whether in the 
common case, i.e., whenever such degenerate configurations are not formed, a stability property still holds.

\subsection{Contribution of this paper}\label{subsec:Contribution} Voronoi diagrams are considered in a class of finite dimensional normed spaces 
having a certain non-exotic unit sphere and the following is shown: If the finitely many compact sites satisfy a certain 
``general position'' assumption related to the structure of the unit sphere, if there is a positive lower bound on the distance between them, and if the  discussion is restricted to a compact and convex subset, then the Voronoi cells 
and their bisectors are stable with respect to small changes in their corresponding sites, where the changes (in the sites and the cells) 
are  measured with respect to the Hausdorff distance. Along the way certain topological 
properties of Voronoi cells which are closely related to their 
geometric stability (e.g., elimination of the possibility of ``fat'' bisectors) are discussed. 
 %
 % and a new geometric characterization of equality in the triangle inequality (in any normed space) 
%is introduced.  
%

Roughly speaking, the characterization of the class of normed spaces which is introduced 
here is that their unit sphere has a finite face decomposition: It can have at most one (possibly zero) ``rotund part'' 
and at most finitely many (possibly zero) ``flat parts'', each of them is a nondegenerate closed and convex subset. In particular,  $(\R^m,\ell_1)$ and $(\R^m,\ell_{\infty})$ belong to this class. The general position condition is that no two points from different sites form a line segment which is parallel to a line segment contained in the unit sphere 
(this condition has been discussed, e.g., in the case of motion planning, but it is not so well-known: 
See Section ~\ref{sec:GeneralPosition}). If the sites are points, 
or very small shapes, and they are generated using the uniform distribution, then this  condition holds with high probability (probability 1 for point sites, unless computational errors and the discrete nature of computational devices are 
taken into account). Hence the stability result is mostly useful for this case (which is the common one dealt 
with in the literature). In general, a configuration of sites having an arbitrary form may violate the 
general position condition (and one may argue that the name ``general position'' is 
not so suitable in this case), but sometimes it is not violated. Anyway, the stability result for this 
case may be interesting from the theoretical point of view and may shed some light on possible problems that 
may occur regarding the cells.  

Comparing with the approaches discussed in the references mentioned a few paragraphs above, the approach of 
this paper is a physical-analytical one: No combinatorial difference between the original and the perturbed 
Voronoi cells is measured, but rather it is ensured 
that if the Hausdorff distance between the original and the perturbed sites is small enough, then so is the 
Hausdorff distance between the original and the perturbed Voronoi cells. Since however no explicit bounds 
are given (in contrast to the dimension free ones mentioned in \cite{ReemGeometricStabilityArxiv}), this continuity result is more theoretical in nature, 
analogous with establishing the convergence of a sequence/algorithm to a limit, but without giving 
estimates on the convergence speed. It somewhat resembles  the continuity result of 
Groemer \cite{Groemer} (but significant differences exist: See \cite{ReemGeometricStabilityArxiv} for a 
comparison). 

\subsection{The structure of the paper} In Section \ref{sec:Definitions} basic definitions and notations are presented. 
In Section \ref{sec:FiniteFace}  the concept of finite face decomposition is introduced and discussed. 
The general position condition is discussed in Section \ref{sec:GeneralPosition}. 
The stability theorem is presented in Section \ref{sec:MainResult} and some aspects related to its proof 
are discussed. In Section \ref{sec:CounterExamples} a few relevant counterexamples are presented, 
showing the necessity of the assumptions imposed in the theorem. A brief discussion regarding 
certain topological properties of Voronoi cells related to their geometric stability is 
given in Section \ref{sec:TopologicalProperties}. 
 A few concluding remarks are discussed in Section \ref{sec:Concluding}.  
 The proofs can be found in Section \ref{sec:ProofStability}.

\section{Notation and basic definitions}\label{sec:Definitions}
Throughout the paper, unless otherwise stated, $X$ is a nonempty compact and convex subset of $\wt{X}=\R^m$, $m\geq 2$, with some norm $|\cdot|$. The induced metric is $d(x,y)=|x-y|$. The unit sphere is $S_{\wt{X}}=\{\theta\in \wt{X}: |\theta|=1\}$. 
It is assumed that $X$ is not a singleton, for otherwise everything is trivial. The notation  $[p,x]$ and $[p,x)$ means the closed and  half open line segments connecting $p$ and $x$, i.e., the sets $\{p+t(x-p): t\in [0,1]\}$ and $\{p+t(x-p): t\in [0,1)\}$ respectively. The segment $[p,x]$ is called non-degenerate if its endpoints 
$p$ and $x$ are different. Any (real) line $L$ can be represented 
as $L=\{p+t\theta: t\in \R\}$ where $\theta$ (the direction vector) is any non-zero vector which, after normalization, 
can be assumed to have norm 1, and $p$ is some point on $L$.  A line $L=\{p+t\theta: t\in \R\}$  is parallel to a second  line $L'=\{p'+t'\theta': t'\in \R\}$ 
if it is a translation of it by some vector, i.e., $L=q+L'=\{q+v': v'\in L'\}$ for some vector 
$q$. Elementary manipulations show that this is equivalent to either $\theta=\theta'$ or $\theta=-\theta'$. 
Two non-degenerate segments $[p,a]$ and $[p',a']$ are said to be parallel if they are located on parallel lines, and this 
is equivalent to either $(a-p)/|a-p|=(a'-p')/|a'-p'|$ or $(a-p)/|a-p|=-(a'-p')/|a'-p'|$. 

%
%We note that the notions and definitions given in this section hold almost word for word whenever $(X,d)$ is any metric 
%space, and the discussion given in Section \ref{sec:FiniteFace} holds for any normed space. 
%

\begin{defin}\label{def:dom}
Given two nonempty subsets $P,A\subseteq X$, the dominance region
$\dom(P,A)$ of $P$ with respect to $A$ is the set of all $x\in X$
whose distance to $P$ is not greater than their distance to $A$, i.e.,
\begin{equation*}
\dom(P,A)=\{x\in X: d(x,P)\leq d(x,A)\}, 
\end{equation*}
where $d(x,A)=\inf\{d(x,a): a\in A\}$. 
\end{defin}
\begin{defin}\label{def:Voronoi}
Let $K$ be a set of at least 2 elements (indices). Given a  tuple $(P_k)_{k\in K}$ of nonempty subsets
$P_k\subseteq X$, called the generators or the sites, the Voronoi diagram  induced by this tuple is 
the tuple $(R_k)_{k\in K}$ such that for all $k\in K$,
\begin{equation*}
R_k=\dom(P_k,{\underset{j\neq k}\bigcup P_j})
=\{x\in X: d(x,P_k)\leq d(x,P_j)\,\,\forall j\in K ,\, j\neq k \}.
\end{equation*}
 In other words, the Voronoi cell $R_k$ associated with the site $P_k$ is the set of all $x\in X$ whose 
 distance to $P_k$ is not greater than their distance to the other sites $P_j, j\neq k$. 
\end{defin}

\begin{defin}\label{def:Hausdorff}
Given two nonempty sets $A_1,A_2\subseteq X$, the Hausdorff distance between them is defined by
\begin{equation*}
D(A_1,A_2)=\max\{\sup_{a_1\in A_1}d(a_1,A_2),\sup_{a_2\in A_2}d(a_2,A_1)\}.
\end{equation*}
\end{defin}
Note that the Hausdorff distance $D(A_1,A_2)$ is definitely different from the usual distance 
\begin{equation*}
d(A_1,A_2)=\inf\{d(a_1,a_2): a_1\in A_1,\,a_2\in A_2\}. 
\end{equation*}
As a matter of fact, $D(A_1,A_2)\leq \epsilon$ if and only if  $d(a_1,A_2)\leq \epsilon$ for all $a_1\in A_1$ and $d(a_2,A_1)\leq \epsilon$ for all $a_2\in A_2$. In addition, if  $D(A_1,A_2)<\epsilon$, then for all $a_1\in A_1$ there exists $a_2\in A_2$ such that $d(a_1,a_2)<\epsilon$ and for all $b_2\in A_2$ there exists $b_1\in A_1$ such that $d(b_2,b_1)<\epsilon$. These properties can be used for giving an optical-geometrical explanation 
for the use of Hausdorff distance as a natural tool when discussing approximation and stability in the context of  sets:  Suppose that our resolution is at most $r$, i.e., we are not able to distinguish between two points whose distance is at most some given positive number $r$. If it is known that $D(A_1,A_2)<r$, then we cannot distinguish between the sets $A_1$ and $A_2$, at least not by inspections based only on distance measurements. As a result of the above discussion, the intuitive phrase ``two sets have almost the same shape'' can be formulated precisely: The Hausdorff distance between the sets is smaller than some given positive parameter. In this context, note that a set and a rigid transformation of it usually have different shapes. In addition, note again that only distance measurements are taken into account in the above discussion, so 
sets which may look different according to other measurements (e.g., differentials properties) but are too close 
in terms of the Hausdorff distance, are considered as equal, because the resolution parameter of the 
involved devices (eyes, magnifying glass, screen, etc.) is too coarse for distinguishing between 
the sets, and in particular, other types of measurements, such as those based on differentials properties, 
will not help. 

\section{Finite face decomposition}\label{sec:FiniteFace}
This section introduces and discusses the seemingly new notion of finite face decomposition, 
a notion which plays an important role in the stability theorem (Section \ref{sec:MainResult}). 
The setting here (and there) is $\R^m$, but the definition can be extended 
word for word to any real or complex normed space and even to vector spaces having a topology. 
It is worthwhile to note that many of the works in the computational geometry 
literature dealing with convex distance functions, e.g., \cite{BSTY1998,ChewDrysdale,CKSTW1998,Drysdale1990,IKLM1995,KLN2009,Le1997,LevenSharir1987,Ma2000,WWW1987}, 
consider unit spheres/circles  which have a finite face decomposition. 
\begin{defin}\label{def:FiniteFace}
The nonempty subset $S\subseteq \R^m$ is said to have a finite face decomposition if $S=\cup_{i=1}^{\ell+1}F_i$ 
where $\ell\in \{0,1,2,3,\ldots\}$ and: 
\begin{enumerate}[(a)]
\item $F_i$ is a closed and convex subset of $S$ for each $i\in \{1,\ldots,\ell\}$ and it is not  
a point; 
\item\label{item:MaxSegment} for each $i\in \{1,\ldots,\ell\}$ the subset $F_i$ has a certain maximality 
property: given any line segment $[s,s']\subseteq S$, if $[s,s'']\subseteq F_i$ for some 
$s''\in (s,s']$, then the whole segment $[s,s']$ is contained in $F_i$;
\item The subset $F_{\ell+1}\subseteq S$ does not contain nondegenerate line segments 
and it satisfies $F_{\ell+1}\cap F_i=\emptyset$ for all $i\in \{1,\ldots, \ell\}$. 
\end{enumerate}
The subsets $F_i$, $i\in \{1,\ldots,\ell\}$  are called the flat parts and $F_{\ell+1}$ is 
called the rotund part. If $S$ is the unit sphere, then $\R^m$ with 
the induced norm is said to have a finite face decomposition property. 
\end{defin}

\begin{example}\label{ex:UnitSphere2D}
As is well-known, the unit sphere of $\R^m$ with the $\ell_{\infty}$ norm, namely with  $\|(x_1,\ldots,x_m)\|_{\infty}=\max\{|x_i|: i=1,\ldots,m\}$), is a cube. Since it can be decomposed into $\ell=2m$ rectangular faces, it has a finite face decomposition. 
Here $F_{\ell+1}=\emptyset$. By the same reasoning the unit sphere of $\R^m$ with the $\ell_{1}$ norm  $\|(x_1,\ldots,x_m)\|_{1}=\sum_{i=1}^m |x_i|$ has a finite face decomposition. The standard Euclidean sphere has a trivial finite face decomposition 
where $\ell=0$ and $F_{\ell+1}$ is the sphere itself (because it does not contain any nondegenerate line segment; 
so is the case in any strictly convex norm space such as the $\ell_p$ spaces, $1<p<\infty$). 
In general, any $m$-dimensional compact polyhedron in $\R^m$ which is symmetric with respect to 
the origin and the origin is an interior point of it (with respect to the Euclidean norm) induces a new and equivalent 
norm on $\R^m$  whose unit sphere is the boundary of this polyhedron. This unit sphere has a finite face decomposition 
with $\ell$=the number of faces, $F_{\ell+1}=\emptyset$. See Figure \ref{fig:UnitSphere2D} for several illustrations. 
\end{example}

\begin{example}\label{ex:UnitSphereStrange3D}
Consider the unit sphere of the Euclidean $\R^3$ cut by the planes $x_3=\alpha$ and $x_3=-\alpha$, 
 where $\alpha\in (0,1)$  is fixed. Let $S=F_1\cup F_2\cup F_3$  where 
\begin{equation*}
F_1=\{(x_1,x_2,x_3)\in\R^3: x_1^2+x_2^2+x_3^2\leq  1,\,x_3=\alpha\}, 
\end{equation*}
\begin{equation*}
F_2=\{(x_1,x_2,x_3)\in\R^3: x_1^2+x_2^2+x_3^2 \leq 1,\,x_3=-\alpha\},  
\end{equation*}
\begin{equation*}
F_3=\{(x_1,x_2,x_3)\in\R^3: x_1^2+x_2^2+x_3^2=1,\,x_3\in (\alpha,-\alpha)\}. 
\end{equation*}
See Figure \ref{fig:UnitSphereStrange3D} for an illustration. This is a finite face decomposition of $S$. 
Since $S$ is symmetric with respect to the origin and the origin is an interior point 
of the body induces by $S$, it follows that $S$ induces a new and equivalent 
norm on $\R^3$, and the corresponding unit sphere is $S$. This unit sphere has a 
finite face decomposition.
\end{example}

%%%%%%%%%%%%%%%%%%%%%%%%%%%%%%%%%%%%%%%%%%%%%%%%%%%%%%%%%%%%%%%%%%%%
\begin{figure}[t]
\begin{minipage}[t]{0.45\textwidth}
\begin{center}
{\includegraphics[scale=0.6]{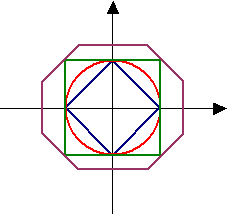}}%{alg40.png}}
\end{center}
 \caption{Several 2D unit spheres with a finite face decomposition: $\ell_1$ (rhombus), $\ell_2$ (circle), $\ell_{\infty}$ 
 (square), something else (octagon).}
\label{fig:UnitSphere2D}
\end{minipage}
\hfill
\begin{minipage}[t]{0.48\textwidth}
\begin{center}
{\includegraphics[scale=0.6]{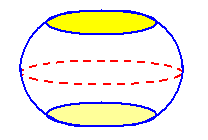}}
\end{center}
 \caption{The unit sphere described in Example \ref{ex:UnitSphereStrange3D}.}
\label{fig:UnitSphereStrange3D}
\end{minipage}
\end{figure}
%%%%%%%%%%%%%%%%%%%%%%%%%%%%%%%%%%%%%%%%%%%%%%%%%%%%%%%%%%%%%%%%%%%%

\begin{example}\label{ex:infinite}
For each $n\in\N$ define the pair $(x_n,y_n)\in \R^2$ by $x_n=1/n$ and $y_n=1-(1/n)$. 
Now connect by a line segment each of the pairs 
$(x_n,y_n)$ and $(x_{n+1},y_{n+1})$. Do the same with the pairs $(-x_n,-y_n)$ and $(-x_{n+1},-y_{n+1})$. 
Finally, connect $(x_{\infty},y_{\infty})=(0,1)$ to $(-x_{1},y_{1})=(-1,0)$ and also 
$(x_{\infty},-y_{\infty})=(0,-1)$ to $(x_1,y_1)=(1,0)$. Let $S$ be the union 
of all of these line segments. The set $S$ is a polygon with infinitely many edges (all of them 
are contained in the  first and the third quadrants) and its vertices are the pairs $(\pm x_n,\pm y_n)$, $n\in \N\cup \{\infty\}$. 
It induces a new unit sphere in the plane $\R^2$ which 
does not have any finite set decomposition. Intuitively, this is because the sets $F_i$ should be the line segments, 
but there are infinitely many line segments. 
\end{example}

\begin{example}\label{ex:cylinder}
The ``cylinder with covers'' $S=\{(x_1,x_2,x_3): |x_1|^2+|x_2|^2=1,\, x_3\in [-1,1]\}$ does not have a finite 
face decomposition. Intuitively, this is because the cylinder contains infinitely many (vertical) line segments 
and each of them must be a ``flat'' face. 
\end{example}

\section{The ``general position'' condition}\label{sec:GeneralPosition}
The goal of this section is to discuss further the ``general position'' assumption which the sites 
$(P_k)_{k\in K}$ should satisfy in order to allow the cells to be geometrically stable with respect to 
small changes of the sites. The precise condition is stated as follows: 

\begin{condition}\label{cond:ParallelLine}
For all $j,k\in K$, $j\neq k$ and for all $p_j\in P_j, p_k\in P_k$, $p_j\neq p_k$, 
the nondegenerate line segment $[p_j,p_k]$ is not parallel to any nondegenerate line segment 
contained in the unit sphere of the space. 
\end{condition}
%%%%%%%%%%%%%%%%%%%%%%%%%%%%%%%%%%%%%%%%%%%%%%%%%%%%%%%%%%%%%%%%%%%%
\begin{figure}[t]
\begin{minipage}[t]{0.45\textwidth}
\begin{center}
{\includegraphics[scale=0.43]{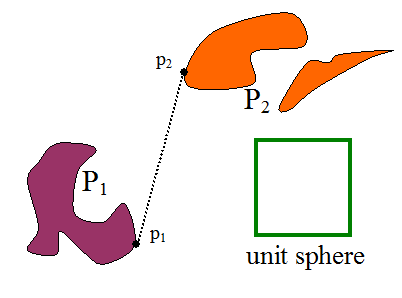}}
\end{center}
 \caption{Illustration of Condition \ref{cond:ParallelLine}: two sites in $(\R^2,\ell_{\infty})$.} 
\label{fig:ParallelLine}
\end{minipage}
\end{figure}
%%%%%%%%%%%%%%%%%%%%%%%%%%%%%%%%%%%%%%%%%%%%%%%%%%%%%%%%%%%%%%%%%%%%%%%
Condition ~\ref{cond:ParallelLine} may perhaps seem somewhat complicated at  first glance, 
but in our opinion this is not the case. It simply says that the sites should be located 
in  a certain way which takes into account the structure of the unit sphere: 
When connecting points from different sites by a line segment, this  segment should not be 
parallel to a line segment contained in the unit sphere. See Figure ~\ref{fig:ParallelLine} 
for an illustration of the condition in the case of $\R^2$ with the $\ell_{\infty}$ norm (and two sites). 
Nevertheless, it may happen that an equivalent (and perhaps even simpler) condition exists. 

It is interesting to note that a version of Condition ~\ref{cond:ParallelLine} appears in the 
study of motion planning of a convex two-dimensional polygonal robot \cite[p. 11, Assumption ~(b)]{LevenSharir1987}.  
It says that no boundary edge of the robot (considered as the unit ball) is parallel to a boundary edge of a 
convex polygonal obstacle (site) or to a line joining a pair of boundary corners of these obstacles.  
 In Remark (2) later it is said that Assumption ~(b) is required for ensuring that the induced Voronoi diagram 
(which, here, is identified with the bisectors, namely with the sets $B_k=\{x\in X: d(x,P_k)=d(x,\cup_{j\neq k} P_j), \}$) 
will be one-dimensional. When this condition is violated then degenerate 
configurations can occur but they can be handled by making infinitesimal perturbations to the sites 
(small rotations). For the purpose of the analysis of a 
related algorithm and the induced diagram, additional assumptions are imposed.
The discussion in \cite[p. 11]{LevenSharir1987} is definitely related to 
stability. However, it is very brief and the focus is more combinatorial or topological 
(the effect of perturbation on topological properties of the bisectors) rather than geometric 
(the effect of perturbation on the shape of the cells or the shape of 
the bisectors). Another related version of Condition ~\ref{cond:ParallelLine}, similar 
to Assumption (b) mentioned above but in $\R^3$, can be found  in \cite[p. 86]{KoltunSharir2004} 
in the context of studying the combinatorial complexity of 3-dimensional Voronoi diagrams 
induced by a polyhedral distance function (the unit ball is 
a polyhedron) and polyhedral sites. 

In fact, since such issues are an integral part of a realistic analysis of the setting, 
the results of this paper may be found applications to \cite{LevenSharir1987}. It seems 
that one such an application is to weaken Assumption ~(b) by removing its first part. Another 
application is Corollary \ref{cor:StabilityBisectors} (Section ~\ref{sec:MainResult}) which ensures that the bisectors 
are stable with respect to small perturbations of the sites. On the other hand, the second part of 
Assumption ~(b) in \cite[p. 11]{LevenSharir1987} is weaker than Condition ~\ref{cond:ParallelLine} 
here unless the sites are points. Hence it is interesting whether Condition ~\ref{cond:ParallelLine} 
can be weakened so that, say, only extreme points $p_j\in P_j$ and $p_k\in P_k$, $j\neq k$  
will be taken into account ($p\in C$ is an extreme point of a set $C$ whenever it cannot be 
written as a strict convex combination of points from $C\backslash \{p\}$; 
in particular, if $C$ is a convex polytope, then its corners are the extreme points).

In general, it seems that  Condition ~\ref{cond:ParallelLine}  is not mentioned frequently and is 
not easily found in the literature, but related variations of it in the case of $(\R^2,\ell_1)$ and 
two point sites are more familiar \cite[p. 390, Figure 37]{Aurenhammer}, \cite[pp. 19-20]{Klein}, \cite[p. 605, Fig. 1(b)]{Lee}, \cite[p. 191, Figure 3.7.2]{OBSC}. In this case either $P_1=\{(-1,-1)\}$ and $P_2=\{(1,1)\}$ or $P_1=\{(-1,1)\}$ and $P_2=\{(1,-1)\}$ and hence the segment $[p_1,p_2]$ generated by the (unique) 
points $p_1=(-1,-1)\in P_1$ and $p_2=(1,1)\in P_2$ is parallel to two segments contained in the unit sphere 
(see Figure \ref{fig:UnitSphere2D} and see also \cite[p. 22]{JKTtech1994} for a closely related example in 
$(\R^2,\ell_{\infty})$ and two point sites). As in \cite[p. 11]{LevenSharir1987}, it 
is mentioned there  that the bisectors induced by this configuration 
are problematic and exotic because they are ``fat'' and that in the typical 
situation (where the sites are not aligned in the above mentioned configuration) the bisectors 
are well-behaved. Some attempts are made in order to 
overcome the problem of fat bisectors by, say, redefining the cells so they will not contain 
the fat bisectors (in this paper the definition  of Voronoi cells, namely Definition ~\ref{def:Voronoi}, 
is not modified). A more detailed discussion about this issue can be found in Section ~\ref{sec:TopologicalProperties}. 

%
%More recently (page 261 of the conference version of \cite{ReemGeometricStabilityArxiv}) it was mentioned 
%that a stability property may hold in, e.g., $(\R^m,\ell_{\infty})$ if the sites satisfy a certain 
%geometric condition, but not enough details were given and neither Condition \ref{cond:ParallelLine} 
%nor any relation to the unit sphere of the space were explicitly mentioned. 
%
 
As a final remark regarding Condition ~\ref{cond:ParallelLine}, it should be emphasized that the ``general position'' 
property required by this condition is significantly different from the common ``general position'' condition 
frequently found in the computational geometry literature dealing with Voronoi diagrams. 
In its simplest form, where the setting is the Euclidean plane,  
this condition says that no 3 distinct (point) sites are located on the same line and no 4 sites are 
located on the same circle. In contrast, in Condition ~\ref{cond:ParallelLine} there is no 
problem at all in both cases even if only point sites are considered. For instance, 3 point sites that 
are located on the same line will not induce instability as long as this line is not parallel to a line 
segment contained in the unit sphere of the space. 

\section{The stability theorem and some aspects related to its proof}\label{sec:MainResult}
In this section the stability theorem and a related corollary are formulated and issues 
related to their proof are briefly discussed. 
For a simple  illustration of the theorem, see Figures \ref{fig:2D-L0}-\ref{fig:2D-L0Pertube}.

\begin{thm}\label{thm:stabilityNUC}
Let $X$ be a compact and convex subset of $(\R^m,|\cdot|)$, $2\leq m\in \N$. 
Assume that the unit sphere of the space has a finite face decomposition. 
Let $K$ be a finite set of indices and let $(P_k)_{k\in K}$  be a given finite tuple 
of nonempty compact  subsets of $X$. For each $k\in K$ let  $A_k=\bigcup_{j\neq k}P_j$ 
and let $R_k=\dom(P_k,A_k)$ be the Voronoi cell corresponding to $P_k$.  Suppose  that 
\begin{equation}\label{eq:eta}
\eta:=\min\{d(P_k,P_j): j,k\in K,\,j\neq k\}>0.
\end{equation}
Suppose also that the sites are located in a ``general position'' with respect to the unit sphere, namely, 
that Condition \ref{cond:ParallelLine} holds. Then for each $\epsilon>0$  there 
exists $\Delta>0$  such that if  $(P'_k)_{k\in K}$ is any tuple of nonempty compact  
subsets of $X$ satisfying  $D(P_k,P'_k)<\Delta$ for each $k\in K$, then $D(R_k,R'_k)<\epsilon$ 
for each $k\in K$. 
\end{thm}
\begin{cor}\label{cor:StabilityBisectors}
Under the assumptions of Theorem \ref{thm:stabilityNUC}, the bisectors 
$B_k:=\{x\in X: d(x,P_k)=d(x,A_k)\}$ are geometric stable: for each $\epsilon>0$  there 
exists $\Delta>0$ (the same one from Theorem \ref{thm:stabilityNUC}) 
such that if  $(P'_k)_{k\in K}$ is any tuple of nonempty compact  
subsets of $X$ satisfying  $D(P_k,P'_k)<\Delta$ for each $k\in K$, then $D(B_k,B'_k)\leq\epsilon$ 
for each $k\in K$. 
\end{cor}

%%%%%%%%%%%%%%%%%%%%%%%%%%%%%%%%%%%%%%%%%%%%%%%%%%%%%%%%%%%%%%%%%%%%
\begin{figure}[t]
\begin{minipage}[t]{0.45\textwidth}
\begin{center}
{\includegraphics[scale=0.59]{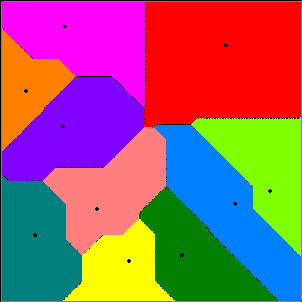}}%{alg40.png}}
\end{center}
 \caption{10 point sites and their cells in a square in $(\R^2,\ell_{\infty})$. 
 Condition \ref{cond:ParallelLine} holds.}
\label{fig:2D-L0}
\end{minipage}
\hfill
\begin{minipage}[t]{0.45\textwidth}
\begin{center}
{\includegraphics[scale=0.59]{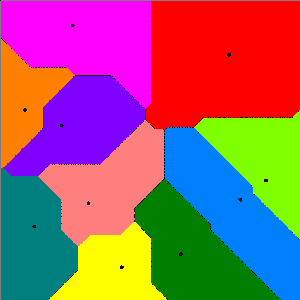}}
\end{center}
 \caption{The sites have slightly perturbed. The shapes of the cells have slightly perturbed.}
\label{fig:2D-L0Pertube}
\end{minipage}
\end{figure}
%%%%%%%%%%%%%%%%%%%%%%%%%%%%%%%%%%%%%%%%%%%%%%%%%%%%%%%%%%%%%%%%%%%%

The proof of Theorem \ref{thm:stabilityNUC} is quite long and technical, and  it is given 
in Section \ref{sec:ProofStability}.  The proof itself is partly inspired by and partly based on arguments and assertions given in  \cite{ReemGeometricStabilityArxiv}. In particular, there is a heavy use of spherical arguments (unit vectors, 
  unit sphere). However, there are considerable differences between both 
proofs and some of the involved  ideas since a key step in \cite{ReemGeometricStabilityArxiv}, namely 
Lemma 5.5 in the extended abstract version and Lemma 8.9 in the current arXiv version (v2), aimed at proving 
a certain geometric estimate, is heavily  based on the assumption  that the normed space is uniformly 
convex. Important tools used here in order to overcome the   difficulty of lacking of this property are 
geometric arguments based on the compactness of the involved sets, arguments based  on finiteness, 
an interesting almost forgotten geometric characterization of equality in  the triangle inequality (instead of the forgotten 
strong triangle inequality of Clarkson \cite[Theorem 3]{Clarkson} used in \cite{ReemGeometricStabilityArxiv}), 
and a certain quantitative way of measuring how far a given vector is from being parallel to a non-degenerate 
segment contained in the unit sphere. The proof of Corollary \ref{cor:StabilityBisectors} 
is a consequence of Theorem \ref{thm:stabilityNUC} and the fact that the bisectors coincide with 
the boundaries of the cells (see Theorem \ref{thm:TopologicalProperties} in Section \ref{sec:TopologicalProperties}). 

The geometric characterization of equality in the 
triangle inequality mentioned above is given below. It has been rediscovered a few 
times but has not become a mainstream knowledge. Apparently it was first considered by 
Golab and H\"arlen \cite{GolabHarlen1931} and later by  Alt \cite{Alt1937}. 
See also \cite[p. 100]{MartiniSwanepoelWeiss2001} for a short and accessible proof and see 
Figure ~\ref{fig:EqualityTriangleInequality} for an illustration. 
\begin{lem}\label{lem:triangle} 
Let $x_1$ and $x_2$ be two non-zero vectors in a normed space. Let $\wh{x_i}=x_i/|x_i|,\,i=1,2$.
Then $|x_1+x_2|=|x_1|+|x_2|$ if and only if the segment $[\wh{x_1},\wh{x_2}]$ is contained in the unit sphere. 
\end{lem}
%%%%%%%%%%%%%%%%%%%%%%%%%%%%%%%%%%%%%%%%%%%%%%%%%%%%%%%%%%%%%%%%%%%%
\begin{figure}[t]
\begin{minipage}[t]{0.45\textwidth}
\begin{center}
{\includegraphics[scale=0.4]{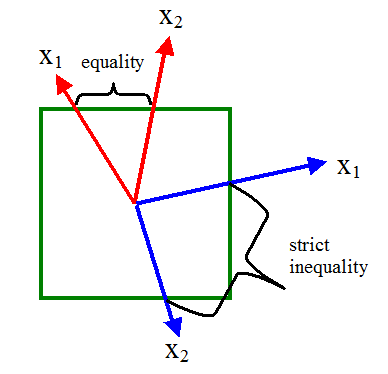}}
\end{center}
 \caption{Illustration of Lemma \ref{lem:triangle} in $(\R^2,\ell_{\infty})$.}
\label{fig:EqualityTriangleInequality}
\end{minipage}
\end{figure}
%%%%%%%%%%%%%%%%%%%%%%%%%%%%%%%%%%%%%%%%%%%%%%%%%%%%%%%%%%%%%%%%%%%

\section{Counterexamples}\label{sec:CounterExamples}

In this section a few counterexamples related to Theorem \ref{thm:stabilityNUC} are discussed. 
They show that the assumptions imposed in the theorem are necessary.

%%%%%%%%%%%%%%%%%%%%%%%%%%%%%%%%%%%%%%%%%%%%%%%%%%%%%%%%%%%%%%%%%%%%
\begin{figure}[t]
\begin{minipage}[t]{0.45\textwidth}
\begin{center}
{\includegraphics[scale=0.56]{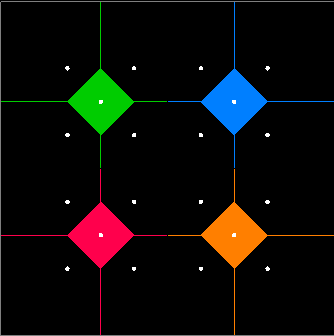}}%{alg40.png}}
\end{center}
 \caption{The figure of Example \ref{ex:2D-L1-20Cell}: 20 point sites in a square in $(\R^2,\ell_1)$. 
 Only four cells are displayed.}
\label{fig:CounterExample2DL1BeforePerturbe}
\end{minipage}
\hfill
\begin{minipage}[t]{0.48\textwidth}
\begin{center}
{\includegraphics[scale=0.56]{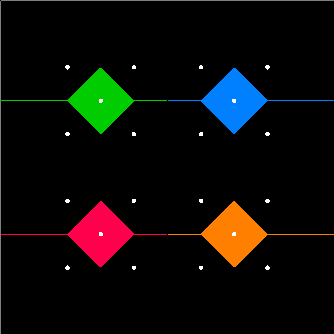}}
\end{center}
 \caption{The non-center sites have slightly moved (in the first component) towards the center site. 
 The vertical rays of the center cells have disappeared. }
\label{fig:CounterExample2DL1AfterPerturbe}
\end{minipage}
\end{figure}
%%%%%%%%%%%%%%%%%%%%%%%%%%%%%%%%%%%%%%%%%%%%%%%%%%%%%%%%%%%%%%%%%%%%

\begin{example}\label{ex:2D-L1-20Cell}
If the sites do not satisfy the ``general position'' assumption (Condition \ref{cond:ParallelLine}), 
i.e., some of them form line segments parallel to segments contained in the unit sphere, then the Voronoi cells 
may not be stable as shown in Figures  \ref{fig:CounterExample2DL1BeforePerturbe}-\ref{fig:CounterExample2DL1AfterPerturbe}. Here the setting is $\R^2$ with the $\ell_1$ norm. There are 20 point sites in the rectangle  $X=[-10,10]\times [-10,10]$. The sites form a symmetric structure composed of 4 groups of 5 sites in each group. The bottom left group consists of $P_1=\{(-6,-6)\},P_2=\{(-2,-6)\}, P_3=\{(-6,-2)\}, P_4=\{(-2,-2)\}$ and the center site $P_5=\{(-4,-4)\}$; the other sites are obtained by translating this group with the vectors $(8,0)$, $(0,8)$, $(8,8)$. Only the cells of $P_5,P_{10}, P_{15}, P_{20}$ are displayed. A small perturbation of the other sites so that they will be closer in the first component to the center site 
of the group by some arbitrary small $\beta>0$ causes the vertical ``rays'' of these sites  to disappear, but 
it preserves the horizontal rays (these rays are, actually, more or less very thin strips). 
\end{example}

\begin{example}\label{ex:LowerBound}
The positive lower bound expressed in \eqref{eq:eta} is necessary even in a square in the plane. Consider $X=[-10,10]^2$, the $\ell_1$ distance,  $P_{1,\beta}=\{(0,\beta)\}$ and $P_{2,\beta}=\{(0,-\beta)\}$, where $\beta\in [0,1]$. As long as $\beta>0$, the cell  $\dom(P_{1,\beta},P_{2,\beta})$ is the upper half of $X$. However, if   $\beta=0$, then $\dom(P_{1,0},P_{2,0})$ is $X$. This example is not specific to non-uniformly convex spaces and can be formulated 
in the Euclidean plane too, as actually done in \cite{ReemGeometricStabilityArxiv}. 
\end{example}
\begin{example}\label{ex:Compactness}
As for the compactness assumption, let $X$ be the (non-compact) plane with the $\ell_{\infty}$  norm and let 
 $P_1=\{(0,10n): n\in \N\}$, $P_2=\{(1/n,5+10n): n\in \N\}$. It can be easily verified that 
 Condition ~\ref{cond:ParallelLine} is satisfied. 
 In addition, $d(P_1,P_2)\geq 5>0$. However, given $\epsilon>0$, there can be no $\Delta>0$ such  that for 
 each $(P'_1,P'_2)$ of nonempty subsets, the inequalities $D(P_1,P_1')<\Delta$, $D(P_2,P'_2)<\Delta$ imply 
 $D(R_1,R'_1)<\epsilon$ and $D(R_2,R'_2)<\epsilon$. Indeed, assume for a contradiction that such $\Delta>0$ 
 exists. Let $n_0\in \N$ be large enough such that $2/n<\Delta$ for each $n\geq n_0,\,n\in\N$. 
 Let $P'_1=\{(0,10n): n<n_0,\,n\in\N\}\cup\{(2/n,10n): n\geq n_0, n\in \N\}$ and $P'_2=P_2$. 
 Then $D(P_1,P'_1)\leq 2/n_0<\Delta$ and $0=D(P_2,P'_2)<\Delta$. However, $D(R'_1,R_1)=\infty$ since the strip 
 $S=\{(x_1,x_2): x_1\leq -10,\, x_2\geq 10n_0\}$ is contained in $R_1=\dom(P_1,P_2)$ but its intersection with 
 $R'_1=\dom(P'_1,P'_2)$ is  empty. 
 
The example described above may seem somewhat complicated, but attempts to construct simpler ones may 
be futile:  for instance, the simple setting where  
$(X,d)=(\R^2,\ell_{\infty})$, $P_1=\{(0,1)\}$, $P'_1=\{(\beta,1\}$, $P_2=\{(0,-1)\}$,  $P'_2=P_2$, 
and  $\beta>0$ is arbitrarily small, imply that $D(R_1,R'_1)=\infty$ (because the negative 
part of the horizontal axis belongs to $R_1$ and not to $R_2$) and there is no stability. 
However, in this case  $P_1$ and $P_2$ are on a line segment parallel to a line segment contained in 
the unit sphere of the space. 
\end{example}

\section{Topological properties of the cells}\label{sec:TopologicalProperties}
This section discusses briefly a few topological properties 
of Voronoi cells under the assumption that Condition \ref{cond:ParallelLine} 
holds. These properties are closely related to the stability theorem (Theorem~\ref{thm:stabilityNUC}) 
and several common ingredients appear in the proofs of these assertions. However, at the moment 
it is not clear whether Theorem ~\ref{thm:TopologicalProperties} below implies 
Theorem ~\ref{thm:stabilityNUC}  or vice versa. 

In what follows $\partial (S)$, $\Int(S)$, and $\overline{S}$ denote the boundary, interior and 
closure of the set $S$, respectively. Note that the condition that the distance between 
each point in the space and the sites is attained implies that the sites are closed. However, 
this could be assumed in advance since a simple verification shows that $\dom(P,A)=\dom(\overline{P},\overline{A})$ 
holds in general. 
\begin{thm}\label{thm:TopologicalProperties}
Let $X$ be a convex subset of an arbitrary normed space.  
Let $P,A\subseteq X$ be nonempty. 
Suppose that $P\bigcap A=\emptyset$ and that the distance between 
each $x\in X$ and both $P$ and $A$ is attained. Suppose also that 
Condition \ref{cond:ParallelLine} holds (with respect to $P_1=P$ and $P_2=A$). 
Then 
\begin{equation}
\partial(\dom(P,A))=\{x\in X: d(x,P)=d(x,A)\},\label{eq:boundary}
\end{equation}
\begin{equation}
\Int(\dom(P,A))=\{x\in X: d(x,P)<d(x,A)\},\label{eq:interior} 
\end{equation}
\begin{equation}
\dom(P,A)=\overline{\{x\in X: d(x,P)<d(x,A)\}}.\label{eq:closure}
\end{equation}
\end{thm}
Property \eqref{eq:boundary} ensures that no ``fat'' bisectors can exist whenever 
Condition ~\ref{cond:ParallelLine} holds. Indeed, since in the context of Voronoi 
cells the bisector of the cell of $P=P_k$ is the set 
$B_k:=\{x\in X: d(x,P_k)=d(x,A_k)\}$ where $A_k=\bigcup_{j\neq k}P_j$, a fat bisector 
means that $\{x\in X: d(x,P)=d(x,A)\}$ contains a ball, a contradiction to \eqref{eq:boundary}. 

It seems that the properties mentioned in Theorem ~\ref{thm:TopologicalProperties} have been essentially 
known for a long time, at least in the 2D case with certain sites (e.g., points or convex polygons) 
or in the case of point sites in certain (finite) higher dimensional spaces. 
See also the discussion in Section ~\ref{sec:GeneralPosition} and also 
\cite{CorbalanMazonRecio1996,Horvath2000,HorvathMartini2013,
IKLM1995,Le1997,LevenSharir1987,Ma2000,MartiniSwanepoel2004} 
where related properties of bisectors (either topological or combinatorial) are 
discussed. However, it is not easy to find explicit formulations of \eqref{eq:boundary}-\eqref{eq:closure} 
in the form given here (in fact, no such formulations have been found) and to the best of the 
author knowledge, these properties have neither been considered in the generality 
mentioned here nor have been proved with the exception of one case: the recent 
established assertion \cite[Lemma 6]{ImaiKawamuraMatousekReemTokuyamaCGTA} which proves 
\eqref{eq:closure} under the assumption that the normed space is finite dimensional 
and strictly convex and the sites are closed and disjoint. 

Very recently \eqref{eq:boundary}-\eqref{eq:closure} have been discussed in additional settings: In 
a class of geodesic metric spaces which contains strictly convex spaces and Euclidean spheres 
\cite{ReemZoneCompute}, assuming the distance to the sites is attained, and in the class of arbitrary 
(possibly infinite dimensional)  
uniformly convex spaces \cite{ReemTopologicalPropertiesPreprint2013} where $d(P,A)>0$ but where 
the distance to the sites may not be attained. See \cite{ReemTopologicalPropertiesPreprint2013} for a 
more extensive discussion about the whole issue. 

\section{Concluding Remarks}\label{sec:Concluding}
The stability property established in Theorem \ref{thm:stabilityNUC} is quite general in the sense that the norm 
is not restricted to be uniformly convex and the sites can have a pretty general form (they only need to be 
compact). However, in some aspects this result is quite restricted comparing to the one presented in  \cite{ReemGeometricStabilityArxiv} since there infinitely many sites were allowed, the sites and the world $X$ were 
not assumed to  be compact or bounded (but a certain  boundedness condition was assumed), and the normed 
space could be even infinite dimensional. Moreover,  explicit (dimension free) estimates were given there in contrast 
with the non explicit ones given here. In particular, it is not clear  whether $\Delta$ given 
in Theorem~ \ref{thm:stabilityNUC} does not depend on the dimension. 
  However, we feel that by  strengthening our  approach it is possible to achieve the above properties in more general normed  spaces, but for this one has to find ways to obtain  explicit estimates not based on compactness or finiteness arguments. 
  In addition, we also feel that the  ``general position'' assumption on the sites can be weakened, with some caution 
  (see e.g., Section \ref{sec:GeneralPosition} for a related discussion). 
It will be interesting to establish results of this spirit. In particular, it will be interesting to weaken the finite face decomposition assumption imposed on the structure of the unit sphere and to see whether 
 a stability result  can be formulated for the normed spaces induced by the unit spheres  mentioned in Examples  \ref{ex:infinite}-\ref{ex:cylinder}.   
 
 It is also interesting to generalize the results to other settings, such 
 as manifolds, weighted distances, and convex distance functions (Minkowski functionals). 
 Finding theoretical and real-world applications 
 of the result (in addition to the possible ones described in Section \ref{sec:Intro}) may be of interest too; 
 a promising place where such a result can be applied is in the context of non-Euclidean stochastic geometry 
 and other kind of probabilistic questions related to geometry and the distribution of 
 the sites, as done in the Euclidean case (see, e.g., \cite{StoyanKendallMecke1987} 
 or \cite[pp. 39, 291-410]{OBSC} for a discussion 
 related to Poisson process). It will be interesting to study the relations between combinatorial 
 stability and geometric stability, e.g., whether there are examples in which the combinatorial structure 
 is stable (there is no combinatorial  difference before and after the perturbation)  but the 
 cells are not geometric stabile since the  Hausdorff distance between the original and perturbed cells 
 is not small enough.   Finally, it may also be interesting 
 to discuss the possibility of stability where there is no one-to-one correspondence between the original sites 
 and the   perturbed ones, e.g., because there were merges or eliminations due to some 
 processes. A corresponding situation occurs when considering point clouds and it may have applications 
 in data analysis and reconstruction, as can be seen in the somewhat related discussion in \cite{ChazalLieutier2005} 
 (in the context of Euclidean lambda-medial axis). 

% 
% \newpage
%

\section{Proofs}\label{sec:ProofStability}

This section establishes the proof of Theorem ~\ref{thm:stabilityNUC}, Corollary ~\ref{cor:StabilityBisectors}, 
and Theorem ~\ref{thm:TopologicalProperties}. In the sequel, unless otherwise stated, $(\widetilde{X},|\cdot|)$ is  
$\R^m$, $2\leq m\in\N$, with some given norm $|\cdot|$; $X$ is a nonempty compact and convex subset of $\wt{X}$ having more than one point; $P,P',A$, and $A'$ are nonempty compact subsets of $X$;  $(P_k)_{k\in K}$ and $(P'_k)_{k\in K}$ are two finite tuples of nonempty compact subsets of $X$ representing the sites and the perturbed ones respectively; for each $k\in K$, we let $A_k=\bigcup_{j\neq k}P_j$ and $A'_k=\bigcup_{j\neq k}P'_j$; unit vectors will usually be denoted by $\theta$ or $\phi$. The unit sphere of the normed space is denoted by $S_{\wt{X}}$. The distance between a set (or a point) and the empty set is defined to be infinity. Note that when the sites are compacts, the distance between a point and a site is 
 attained. In general, we note that some of the claims and definitions given in the sequel hold in a more general setting (e.g., in some of them compactness or finite dimension are not needed; see  \cite{ReemISVD09,ReemGeometricStabilityArxiv}). 
 However, to avoid apparent complication the general case is not always explicitly considered. 
 The proof of Theorem ~\ref{thm:domIntervalApp}  can be found in  \cite{ReemISVD09,ReemGeometricStabilityArxiv}  
 and the proof of Proposition \ref{prop:stabilityVor} can be found in  \cite{ReemGeometricStabilityArxiv}. 
 
\begin{thm}{\bf(The representation theorem)}\label{thm:domIntervalApp}
 $\dom(P,A)$ is a union of line segments starting at the points of $P$. More precisely, given $p\in P$ and a unit vector $\theta$, let
\begin{equation}\label{eq:Tdef}
T(\theta,p)=\sup\{t\in [0,\infty): p+t\theta\in X\,\,\mathrm{and}\,\,
 d(p+t\theta,p)\leq d(p+t\theta,A)\}.
\end{equation}
Then
\begin{equation*}%\label{eq:dom}
\dom(P,A)=\bigcup_{p\in P}\bigcup_{|\theta|=1}[p,p+T(\theta,p)\theta].
\end{equation*}
%When $T(\theta,p)=\infty$, the notation $[p,p+T(\theta,p)\theta]$ means the ray $\{p+t\theta: t\in [0,\infty)\}$.
\end{thm}

\begin{prop}\label{prop:stabilityVor}
Suppose that $\inf\{d(P_k,P_j): j,k\in K,\,j\neq k\}>0$. Let $\epsilon>0$ be such that $\epsilon\leq\inf\{d(P_k,P_j): j,k\in K,\,j\neq k\}/6$. Suppose  that the following conditions hold:
\begin{multline}\label{eq:property_k}
\exists\,\lambda\in (0,\epsilon)\,\, \textnormal{such that for each}\,\, k\in K, p\in P_k,\,\,y\in \dom(P_k,A_k),\, \textnormal{and}\, x\in [p,y]\, \\
\textnormal{if}\,\,d(x,y)=\epsilon/2,\,\, \textnormal{then}\,\,d(x,p)\leq d(x,A_k)-\lambda.
\end{multline}
\begin{multline}\label{eq:property'_k}
\exists\,\lambda'\in (0,\epsilon)\,\, \textnormal{such that for each}\,\, k\in K, p'\in P'_k,\,y'\in \dom(P'_k,A'_k),\,\, \textnormal{and}\,\,x'\in [p',y'] \\
\textnormal{if}\,\, d(x',y')=\epsilon/2,\,\, \textnormal{then}\,\,d(x',p')\leq d(x',A'_k)-\lambda'.
\end{multline}
Suppose that there are  $M,M'\in (0,\infty)$ such that for all $k\in K$,
\begin{equation}\label{eq:Mk}
\sup\{T_k(\theta,p): p\in P, |\theta|=1\}\leq M,\quad\sup\{T'_k(\theta',p'): p'\in P',|\theta'|=1\}\leq M',
\end{equation}
where $T_k(\theta,p)$ and $T'_k(\theta',p')$ are defined as in \eqref{eq:Tdef}, but with $A_k$ and $A'_k$ instead of $A$. 

Let  $\epsilon_4$ be a positive number satisfying
\begin{equation}\label{eq:epsilon24M}
4(1+M/\epsilon)\epsilon_4<\lambda/2, \quad 4(1+M'/\epsilon)\epsilon_4<\lambda'/2.
\end{equation}
Let
\begin{equation*}\label{eq:R_k}
R_k=\dom(P_k,A_k), \quad
\quad R'_k=\dom(P'_k,A'_k).
\end{equation*}
If
\begin{equation}\label{eq:epsilon24k}
  \quad D(P_k,P'_k)< \epsilon_4 \,\,\,\,\forall k\in K,
\end{equation}
 then $D(R_k,R'_k)<\epsilon$ for each $k\in K$.
\end{prop}

%\subsection{Several technical tools}
\begin{defin}\label{def:[P,A]}
Given two nonempty subsets $P$ and $A$ of $X$, the notation $[P,A]$ is 
the set of all non-degenerate line segments of the form $[p,a]$, where
$p\in P$, $a\in A$. The notation $\widehat{[P,A]}$ is for the set of all unit 
vectors generated  by segments from $[P,A]$. The notation $\widehat{P,A}$ means 
 the set of all unit vectors generated  by endpoints of segments from $[P,A]$. 
In other words,
\begin{equation}\label{eq:[P,A]}
\widehat{[P,A]}=\left\{\frac{a-p}{|a-p|}: [p,a]\in [P,A]\right\}. 
\end{equation}
\begin{equation}\label{eq:|P,A|}
\widehat{P,A}=\widehat{[P,A]}\cup\widehat{[A,P]}. 
\end{equation}
The notation $\wh{A}$ is the set of unit vectors generated from nondegenerate line segments contained 
in $A$, namely 
\begin{equation}\label{eq:UnitVectorLineSegment}
\wh{A}=\left\{\frac{a'-a}{|a'-a|}: [a,a']\subseteq A, a\neq a'\right\}. 
\end{equation}
\end{defin}
The most important case for \eqref{eq:UnitVectorLineSegment}  is when $A=S_{\wt{X}}$, i.e., 
when $A$ is the unit sphere. If for instance the normed space is 
$(\R^2,\ell_{\infty})$, then $\wh{S_{\wt{X}}}=\{(0,1), (0,-1),(1,0), (-1,0)\}$. 
 In this case $\wh{S_{\wt{X}}}$ is compact, but in general it may not be, and this is the case 
 of $\wh{S_{\wt{X}}}$ from Example \ref{ex:UnitSphereStrange3D}. If the unit sphere does not 
 contain any line segment (i.e., the normed space is strictly convex), then $\wh{S_{\wt{X}}}=\emptyset$. 
%
%is given in Figure \ref{fig:UnitVectorLineSegmentLinfty}.
%

The following simple lemma (whose proof is given for the sake of completeness) 
establishes  simple properties of $\widehat{P,A}$.
\begin{lem}\label{lem:wh[P,A]}
\begin{enumerate}[(I)]
\item\label{item:parallel} Given $x,y\in X$, $x\neq y$, the segment $[x,y]$ is parallel to some segment contained in $[P,A]$ if and only if $(y-x)/|y-x|\in \wh{P,A}$, i.e., if and only if $\pm(y-x)/|y-x|\in \wh{[P,A]}$. 
\item\label{item:HSwh[P,A]} Let $H=\{t(a-p): a\in A, p\in P, t\in \R\}$. 
Then $\widehat{P,A}=H\cap S_{\wt{X}}$. 
\end{enumerate}
\end{lem}
\begin{proof}
Part \eqref{item:parallel} follows directly from the fact that the non-degenerate segments 
$[x,y]$ and $[p,a]$ are parallel if and only if $(y-x)/|y-x|=(a-p)/|a-p|$ or $-(y-x)/|y-x|=(a-p)/|a-p|$. 
Part \eqref{item:HSwh[P,A]} is simple too: if $\theta=\pm(a-p)/|a-p|\in \wh{P,A}$, then $|\theta|=1$ and  $\theta=t(a-p)$ 
for $t=\pm 1/|a-p|$. Hence $\theta\in H\cap S_{\wt{X}}$. On the other hand, if 
$\theta\in H\cap S_{\wt{X}}$, then $\theta=t(a-p)$ for $a\in A, \,p\in P$, and $t\in \R$. Since $|\theta|=1$ we have $a\neq p$ 
and $|t|=1/|a-p|$. Hence either  $t=1/|a-p|$ or  $t=-1/|a-p|$ and thus $\theta\in \wh{P,A}$. 
\end{proof}

The following lemma establishes a simple estimate on the Hausdorff distance of unions of sets in 
terms of the Hausdorff distances of members in the unions. 
\begin{lem}\label{lem:HausdorffUnion}
Let $(G_i)_{i\in I}$ and $(G'_i)_{i\in I}$ be two tuples of subsets of a metric space, 
where $I$ is any nonempty index set. Then 
\begin{equation*}
D(\cup_{i\in I}G_i,\cup_{i\in I}G'_i)\leq \sup\{D(G_i,G'_i): i\in I\}.
\end{equation*}
\end{lem}
\begin{proof}
Let $y\in \cup_{i\in I}G_i$. Then $y\in G_j$ for some $j\in I$. 
Hence 
\begin{multline*}
d(y,\cup_{i\in I}G'_i)\leq d(y, G'_j)\leq \sup\{d(x,G'_j): x\in G_j\}\leq D(G_j,G'_j) 
\leq  \sup\{D(G_i,G'_i): i\in I\}
\end{multline*}
by the definition of the Hausdorff distance $D$. 
Thus $\sup\{d(x, \cup_{i\in I}G'_i): x\in \cup_{i\in I}G_i\}\leq \sup\{D(G_i,G'_i): i\in I\}$. 
In the same way $\sup\{d(x', \cup_{i\in I}G_i): x'\in \cup_{i\in I}G'_i\}\leq \sup\{D(G_i,G'_i): i\in I\}$.  
Therefore $D(\cup_{i\in I}G_i,\cup_{i\in I}G'_i)\leq  \sup\{D(G_i,G'_i): i\in I\}$, again 
by the definition of the Hausdorff distance. 
\end{proof}

The next lemma establishes estimates on $|d(P',A')-d(P,A)|$ and $D(\wh{P,A},\wh{P',A'})$ 
whenever there are known estimates on $D(P,P')$ and $D(A,A')$. 
\begin{lem}\label{lem:whPAP'A'}
Suppose  that $D(P,P')<\epsilon_1$ and $D(A,A')<\epsilon_2$ for some positive 
numbers $\epsilon_1,\epsilon_2$. Then 
\begin{enumerate}[(I)] 
\item\label{item:d(P,A)eps1eps2}  $|d(P',A')-d(P,A)|\leq \epsilon_1+\epsilon_2$
\item\label{item:wh[P,A]} 
If $\epsilon_1+\epsilon_2<d(P,A)$ then 
\begin{equation}\label{eq:D[P,A][P',A']}
D(\wh{[P,A]},\wh{[P',A']})\leq\frac{2(\epsilon_1+\epsilon_2)}{d(P,A)} 
\end{equation}
and 
\begin{equation}\label{eq:D|P,A||P',A'|}
D(\wh{P,A},\wh{P',A'})\leq\frac{2(\epsilon_1+\epsilon_2)}{d(P,A)}.
\end{equation}

\end{enumerate}
\end{lem}

\begin{proof}
We first prove part \eqref{item:d(P,A)eps1eps2}. Let $p'\in P'$ and $a'\in A'$ be given. 
Since  $D(P,P')<\epsilon_1$ and $D(A,A')<\epsilon_2$ there are $p\in P$ and $a\in A$ 
such that $d(p,p')<\epsilon_1$ and $d(a',a)<\epsilon_2$, so
\begin{equation*}%\label{eq:d(P',A')}
d(P,A)\leq d(p,a)\leq d(p,p')+d(p',a')+d(a',a)< \epsilon_1+\epsilon_2+d(p',a'),
\end{equation*}
Since $p'$ and $a'$ were arbitrary it follows that $d(P,A)-\epsilon_1-\epsilon_2\leq d(P',A')$. 
In the same way  $d(P',A')-\epsilon_1-\epsilon_2\leq d(P,A)$.

We now pass to part \eqref{item:wh[P,A]}. It suffices to prove \eqref{eq:D[P,A][P',A']} since in the same 
way an analogous inequality with $\wh{[A,P]},\wh{[A',P']}$ is proved, and using Lemma \ref{lem:HausdorffUnion} 
and $\wh{P,A}=\wh{[P,A]}\cup \wh{[A,P]}$ we obtain \eqref{eq:D|P,A||P',A'|}. 

Since $\epsilon_1+\epsilon_2<d(P,A)$ it follows that $|a-p|>0$ and,  
from Part \eqref{item:d(P,A)eps1eps2}, that $|a'-p'|>0$ whenever $p\in P$, $a\in A$, $p'\in P'$, $a'\in A'$ and 
$d(p,p')<\epsilon_1$, $d(a,a')<\epsilon_2$. Now suppose that $[p,a]\in [P,A]$ is given. 
Since $D(P,P')<\epsilon_1$ there exists 
$p'\in P'$ such that $d(p,p')<\epsilon_1$. Since $D(A,A')<\epsilon_2$ there exists 
$a'\in A'$ such that $d(a,a')<\epsilon_2$. Hence $p'=p+h_1\phi_1$ and $a'=a+h_2\phi_2$ where $h_i\in (0,\epsilon_i)$ and $|\phi_i|=1$, $i=1,2$ (e.g., $\phi_1=(p'-p)/|p'-p|$).  Let $h=h_2\phi_2-h_1\phi_1$ and $u=a-p$. Then $a'-p'=u+h$ and 
\begin{multline}\label{eq:eps1eps2}
\left|\frac{a'-p'}{|a'-p'|}-\frac{a-p}{|a-p|} \right|=\left|\frac{|a-p|(a'-p')-|a'-p'|(a-p)}{|a'-p'||a-p|}\right|\\
=\left|\frac{|u|(u+h)-u|u+h|}{|a'-p'||a-p|}\right| 
= \frac{\left|(u+h)(|u|-|u+h|)+|u+h|(u+h-u)\right|}{|u+h||a-p|}\\
\leq \frac{2|u+h||h|}{|u+h|d(P,A)}
< \frac{2 (\epsilon_1+\epsilon_2)}{d(P,A)}
\end{multline}
where the triangle inequality was used a few times for establishing inequality. 
In the same way as in \eqref{eq:eps1eps2} for all $[p',a']\in [P',A']$ there exists $[p,a]\in [P,A]$ such that 
\begin{equation*}
\left|\frac{a'-p'}{|a'-p'|}-\frac{a-p}{|a-p|}\right|<\frac{2(\epsilon_1+\epsilon_2)}{d(P,A)}. 
\end{equation*}
Therefore we obtain \eqref{eq:D[P,A][P',A']}.
\end{proof}

The following lemma is needed for the proof of Lemma \ref{lem:SegmentInequality} and Lemma \ref{lem:lambda_uni}. 
Now the characterization of equality in the triangle inequality (Lemma ~\ref{lem:triangle}) is used. 
\begin{lem}\label{lem:ParallelUnitSphere}
Let $p$, $y$, and $a\neq p$ be points in an arbitrary normed space satisfying 
$0<d(y,p)\leq d(y,a)$. If $x\in (p,y)$ satisfies $d(x,p)=d(x,a)$, 
then the segment $[(x-a)/|x-a|,(x-p)/|x-p|]$ is nondegenerate and it 
is contained in the unit sphere of the space. In addition, $[p,a]$ is parallel to this 
segment. 
\end{lem}
\begin{proof}
Since $x\in [p,y]$ and $d(y,p)\leq d(y,a)$,  
\begin{equation}\label{eq:xpxa}
|y-x|+|x-p|=|y-p|\leq |y-a|\leq |y-x|+|x-a|. 
\end{equation}
Therefore $|x-p|\leq |x-a|$.   By assumption 
\begin{equation}\label{eq:xp=xa}
|x-p|=|x-a|. 
\end{equation}
This equality and  \eqref{eq:xpxa} imply that 
\begin{equation}\label{eq:yayp}
|y-a|=|y-p|
\end{equation}
and 
\begin{equation}\label{eq:triangle_axy}
|y-a|=|y-x|+|x-a|. 
\end{equation}
It must be that $a$ is not on the line passing through $y$ and $p$. Indeed, in order to be on 
this line there are several cases. In the first case $a$ is beyond $p$ in the direction 
of $p-y$. But then $p\in (y,a)$ and hence $|y-a|>|y-p|$, 
a contradiction to \eqref{eq:yayp}. In the second case $a$ is beyond $y$ in the direction 
of $y-p$. But then in particular $y\in (x,a)$ and therefore $|y-a|<|x-a|$, a contradiction 
to \eqref{eq:triangle_axy}. In the third case $a\in (p,y)$, but then $|y-a|<|y-p|$, 
a contradiction to the assumption $d(y,p)\leq d(y,a)$ in the formulation of the lemma. 
The case $y=a$ cannot hold because otherwise \eqref{eq:triangle_axy} would imply  
that $x=y$, a contradiction. The last case $a=p$ is impossible from  the 
formulation of the lemma. Since $x\neq a$ (otherwise $x=p$) the above implies that 
$(x-a)/|x-a|\neq (y-p)/|y-p|=(x-p)/|x-p|$. It follows that the segment $[(x-a)/|x-a|,(x-p)/|x-p|]$ is non-degenerate and by  \eqref{eq:triangle_axy}, the equality $(y-x)/|y-x|=(x-p)/|x-p|$, and Lemma ~\ref{lem:triangle} this segment is 
contained in the unit sphere. This segment is parallel to $[p,a]$ since \eqref{eq:xp=xa} implies that 
\begin{equation*}
\frac{x-a}{|x-a|}-\frac{x-p}{|x-p|}=\frac{(p-a)}{|x-p|}, 
\end{equation*}
and so 
\begin{equation*}
\frac{((x-a)/|x-a|)-((x-p)/|x-p|)}{|((x-a)/|x-a|)-((x-p)/|x-p|)|}= 
\frac{(p-a)/|x-p|}{|(p-a)/|x-p||}=\frac{p-a}{|p-a|}.
\end{equation*}
 \end{proof}

 The following lemma plays a key role in the proofs of several later assertions, including 
 Theorem ~\ref{thm:TopologicalProperties}. 
\begin{lem}\label{lem:SegmentInequality}
Let $p$ and $y$ be points in a normed space and let $A$ be a nonempty subset contained in this space. 
Suppose that $p\notin A$ and that $d(x,A)$ is attained for all $x\in [p,y)$. Suppose also 
that $d(y,p)\leq d(y,A)$ and that for each $a\in A$ the segment $[p,a]$ is not parallel to any non-degenerate segment contained in the unit sphere of the space. Then $d(x,p)<d(x,A)$ for all $x\in [p,y)$.
\end{lem}
\begin{proof}
If $y=p$, then the assertion is obvious (void).  Now assume that $y\neq p$ and let $x\in [p,y)$. 
If $x=p$, then $d(x,p)<d(x,A)$ because $p\notin A$ and $d(p,A)$ is attained. Assume that $x\in (p,y)$ and 
 let $a\in A$ satisfy $d(x,a)=d(x,A)$. Since $x\in [p,y]$ and $d(y,A)\leq d(y,a)$, it follows from  \eqref{eq:xpxa} 
 (without the need to assume that $d(x,p)=d(x,a)$) that $|x-p|\leq |x-a|$.  Assume by way of contradiction that $d(x,p)=d(x,a)$. 
 Then Lemma ~\ref{lem:ParallelUnitSphere} can be used and it implies that $[p,a]$ is parallel to a nondegenerate line 
 segment which is contained in the unit sphere. This contradicts the assumption in the formulation 
 of the lemma. Thus $d(x,p)<d(x,a)=d(x,A)$ as claimed.
 \end{proof}

The following lemma proves an estimate based on compactness. 
\begin{lem}\label{lem:LambdaNUC}
Suppose that $d(P,A)>0$, that $A$ is compact, and that the following  condition holds:
\begin{multline}\label{eq:InSegment}
\textnormal{for all}\,\, p\in P,\,\,a\in A\,\,\textnormal{and}\,\,x,y\in X,\\\,\,\textnormal{if } \,\,d(y,p)\leq d(y,a)\,\,\textnormal{and}\,\,x\in [p,y),\,\textnormal{then}\,\,d(x,p)<d(x,a). 
\end{multline}
Suppose that $\epsilon>0$ satisfies $\epsilon\leq d(P,A)/2$.  Then  the following condition holds:
\begin{multline}\label{eq:property}
\textnormal{There exists}\,\, \lambda\in (0,\epsilon),\,\, \textnormal{depending only on}\,\,\epsilon\,\,
\textnormal{and the sets}\,\,P,A,\,\,\\ \textnormal{such that for each}\,\, p\in P, y\in \dom(P,A),\,x\in [p,y],\, \\
\textnormal{if}\,\, d(x,y)=\epsilon/2,\,\, \textnormal{then}\,\,d(x,p)\leq d(x,A)-\lambda.
\end{multline}
\end{lem}

\begin{proof}
Define 
\begin{equation}\label{eq:S_lambda_til}
\begin{array}{l}
S=\{(x,p): x\in X,p\in P\,\,\textnormal{and there exists}\,\,y\in X \,\,\textnormal{such that}\\
\quad\quad\quad \,x\in [p,y],\,d(y,p)\leq d(y,A)\,\,\textnormal{and}\,\,d(x,y)=\epsilon/2\},\\
\wt{\lambda}=\inf\{d(x,A)-d(x,p): (x,p)\in S\}.
\end{array}
\end{equation}
Given  $p\in P$, let $b\in A$ be arbitrary. Let $\theta=(b-p)/|b-p|$,  $x=p+0.25\epsilon\theta$, 
and $y=p+0.75\epsilon\theta$. By the convexity of $X$ we have $x,y\in [p,b]\subseteq X$.  Given $a\in A$, 
since $d(P,A)\geq 2\epsilon$ we have  $2\epsilon\leq d(a,p)\leq d(a,y)+d(y,p)=d(a,y)+0.75\epsilon$. Because $a$ was arbitrary we have $1.25\epsilon\leq d(y,A)$. It follows that $d(y,p)\leq d(y,A)$ and $(x,p)\in S$. Thus $S\neq \emptyset$. 

For proving the assertion it suffices to show that $\wt{\lambda}$ is attained at some $(x,p)\in S$. To see why, let $y$ correspond to $x$ and $p$ in the definition of $S$. By assumption $A$ is compact. Hence $d(x,A)=d(x,a)$ for some $a\in A$. From the property of $y$ we have $d(y,p)\leq d(y,A)\leq d(y,a)$. Thus, by \eqref{eq:InSegment} it follows that $\wt{\lambda}=d(x,a)-d(x,p)>0$, and by taking \begin{equation}\label{eq:lambda_epsilon}
\lambda=\min\{\epsilon/2,\wt{\lambda}\} 
\end{equation}
it follows that \eqref{eq:property} indeed holds.

 To show that $\wt{\lambda}$ is attained, consider the sequences $(x_n)_{n=1}^{\infty}$ and $(y_n)_{n=1}^{\infty}$ such that $(x_n,p_n)\in S$  and  $\wt{\lambda}=\lim_{n\to \infty} (d(x_n,A)-d(x_n,p_n))$. For each $n\in \N$ let $y_n$ correspond to 
 $(x_n,p_n)$ in the definition of $S$. We can write $x_n=p_n+t_n\theta_n$ where $|\theta_n|=1$ and $t_n\in [0,\infty)$. In fact, $d(y_n,p_n)\geq d(y_n,x_n)=\epsilon/2$ and also $\theta_n=(y_n-p_n)/d(y_n,p_n)$. In addition, $t_n\leq\textnormal{diam}(X)$. 
 %because 
 %$\begin{equation*}
%t_n=d(x_n,p_n)\leq d(y_n,p_n)\leq d(y_n,A)\leq\rho
%\end{equation*}
 %by \eqref{eq:BallRho}. 
 Since $P$ is compact, we can pass to a subsequence and obtain that $t=\lim_{l\to\infty} t_{n_l}$ and $p=\lim_{l\to\infty} p_{n_l}$ for some $t\in [0,\textnormal{diam}(X)]$ and $p\in P$. Since $X$ is compact, then  also $y=\lim_{i\to\infty} y_{n_{l_i}}$ for some $y\in X$ and a subsequence $(n_{l_i})_{i=1}^{\infty}$ of positive integers, and then  $\theta=\lim_{i\to\infty}\theta_{n_{l_i}}$ for $\theta=(y-p)/d(y,p)$. In addition, $x=\lim_{i\to\infty} x_{n_{l_i}}$ for $x=p+t\theta$. 
 
 %Otherwise  $\wt{X}$ is finite dimensional, so its unit sphere is compact, and then we can find a subsequence %$(n_{l_i})_{i=1}^{\infty}$ of positive integers such that  $\theta=\lim_{i\to\infty}\theta_{n_{l_i}}$ for some unit vector %$\theta$. Thus  $x_{n_{l_i}}\to x=p+t\theta$, and  $y_{n_{l_i}}\to y:=p+(t+0.5\epsilon)\theta$.
 
  As a result,  $x\in [p,y]$. Finally, $d(y,x)=\epsilon/2$ and $d(y,p)\leq d(y,A)$ by the continuity of the distance functions. Thus $(x,p)\in S$ and $\wt{\lambda}=d(x,A)-d(x,p)>0$.
   
\end{proof}

The next lemma establishes a simple property of subsets having a finite face decomposition. 
\begin{lem}\label{lem:[s,s']}
If $S$ has a finite face decomposition $S=\cup_{i=1}^{\ell+1}F_i$, then any non-degenerate line segment $[s,s']$ contained in  $S$ must be contained in $F_i$ for some $i\in \{1,\ldots,\ell\}$. 
\end{lem}
\begin{proof}
First note that $[s,s']\cap F_{\ell+1}=\emptyset$ since otherwise  
the fact that $F_{\ell+1}\cap (\cup_{i=1}^{\ell}F_i)=\emptyset$ implies (using the fact that $F_i$ is closed 
for each $i$) that a small part of $[s,s']$ around any $s''\in [s,s']\cap F_{\ell+1}$ is contained in $F_{\ell+1}$, a contradiction. Now consider $s$: there exists a maximal nonempty subset $I\subseteq\{1,\ldots,\ell\}$ such that $s\in F_i$ for any $i\in I$. If $I=\{1,\ldots,\ell\}$, then because $s'\in F_j$ 
for some $j\in \{1,\ldots,\ell\}$ it follows that both $s,s'\in F_j$ and then, by convexity, $[s,s']\subseteq F_j$. 
Otherwise $I$ is strictly contained in $\{1,\ldots,\ell\}$. 

It must be that some $s''\in (s,s']$ belongs to 
$F_i$ for some $i\in I$. Otherwise there exists a subsequence $(s''_n)_{n=1}^{\infty}$ converging 
to $s$ but satisfying $s_n''\notin F_i$ for any $i\in I$ and $n\in \N$. Because $\{1,\ldots,\ell\}\backslash I$ is 
nonempty and finite it follows that 
infinitely many of these $s''_n$ belong to $F_j$ for some fixed $j\not\in I$. But $F_j$ is closed by assumption, 
so the limit $s$ is in $F_j$ too, a contradiction to the maximality of $I$. Therefore there exists some $s''\in F_i\cap (s,s']$ for some $i\in I$, and by convexity it follows that $[s,s'']\subseteq F_i$. Hence $[s,s']\subseteq F_i$ by the maximality property of $F_i$ (see Property \eqref{item:MaxSegment} in Definition \ref{def:FiniteFace}). 
\end{proof}

The following lemma establishes a simple property of a set which is related to Definition \ref{def:[P,A]}. 
It is probably known and the proof is given for the sake of completeness. 
\begin{lem}\label{lem:SpanF}
Suppose that $F$ is a nonempty convex subset of a real vector space $V$. Let 
\begin{equation*}
H=\{t(v-u): \,t\in \R,\,\,u,v\in F\,\,\textnormal{are arbitrary}\}.
\end{equation*}
Then $H$ is a linear subspace of $V$. In particular, if $V$ is finite dimensional, 
then $H$ is closed with respect to any norm defined 
on $V$. 
\end{lem}
\begin{proof}
Since $0\in H$ and since $H$ is closed under scalar multiplication, it is sufficient to 
show that $H$ is closed under addition. Let $t_1(v_1-u_1),t_2(v_2-u_2)\in H$. The assertion 
is obvious if $t_1=0$ or $t_2=0$. Assume now that both $t_1$ and $t_2$ are positive. Because 
\begin{equation*}
t_1(v_1-u_1)+t_2(v_2-u_2)=(t_1+t_2)\left(\left(\frac{t_1}{t_1+t_2}v_1+\frac{t_2}{t_1+t_2}v_2\right)-
\left(\frac{t_1}{t_1+t_2}u_1+\frac{t_2}{t_1+t_2}u_2)\right)\right)
\end{equation*}
and because $F$ is convex, it follows that $t_1(v_1-u_1)+t_2(v_2-u_2)\in H$. The case where $t_i<0$ 
for some $i$ follows from the previous case by observing that $t_i(v_i-u_i)=(-t_i)(u_i-v_i)\in H$. 

Finally, if in addition $V$ is finite dimensional, then it is well-known 
that all its linear subspaces are closed with respect to any norm defined on it. 
\end{proof}

The next few lemmas establish additional estimates needed in the proof of the stability result or in the 
proof of certain auxiliary assertions. 
\begin{lem}\label{lem:NoParallelPA}
Suppose that $d(P,A)>0$ and any segment $[p,a]\in [P,A]$ is not parallel to a line segment 
contained in the unit sphere $S_{\wt{X}}$ of the space. Suppose also that $S_{\wt{X}}$ has  
the finite face decomposition $S_{\wt{X}}=\cup_{i=1}^{\ell+1}F_i$. Then $d(\wh{P,A},\wh{S_{\wt{X}}})>0$.
\end{lem}
\begin{proof}
The assertion is obvious (void) if $\wh{S_{\wt{X}}}=\emptyset$, since the distance to the empty set 
is infinity. Now assume that $\wh{S_{\wt{X}}}\neq\emptyset$ and suppose by way of contradiction that  $d(\wh{P,A},\wh{S_{\wt{X}}})=0$. 
Then there exist sequences $(\phi_n)_{n=1}^{\infty}$, $(\theta_n)_{n=1}^{\infty}$ 
such that $\lim_{n\to\infty}|\phi_n-\theta_n|=0$, $\phi_n\in \wh{P,A}$, and 
$\theta_n\in \wh{S_{\wt{X}}}$ for each $n\in \N$. By the definition of $\wh{S_{\wt{X}}}$ 
for each $n\in\N$ we have $\theta_n=(s'_n-s_n)/(s'_n-s_n)$ where $[s_n,s'_n]\subseteq S_{\wt{X}}$,  
$s_n\neq s'_n$. 

Since $\wh{P,A}=\wh{[P,A]}\cup\wh{[A,P]}$, there exists an infinite subset $N_0$ of $\N$ 
such that $\phi_n\in \wh{[P,A]}$ for each $n\in \N_0$, or  $\phi_n\in \wh{[A,P]}$ for each $n\in \N_0$. 
Assume the first case. The other case can be considered in the same way.
Then there exist sequence $(p_n)_{n\in N_0}$ and $(a_n)_{n\in N_0}$ 
such that $[p_n,a_n]\in [P,A]$ for each $n\in N_0$ and 
\begin{equation}\label{eq:a_np_ns_ns'_n}
\lim_{n\to\infty, n\in N_0}\left|\frac{a_n-p_n}{|a_n-p_n|}-\frac{s'_{n}-s_{n}}{|s'_{n}-s_{n}|}\right|=0. 
\end{equation}
By passing to subsequences and using the compactness of $P$ and $A$ we conclude that 
there exist $p\in P$ and $a\in A$ such that $p=\lim_{n\to \infty,n\in N_1}p_n$ and $a=\lim_{n\to \infty,n\in N_1}a_n$ 
for an infinite subset $N_1$ of $N_0$. By Lemma \ref{lem:[s,s']} for each $n\in N_1$ there exists $i\in \{1,\ldots,\ell\}$ such that  $[s_n,s'_n]\subseteq F_i$. Because there are finitely many subsets $F_i,\,i\in \{1,\ldots,\ell\}$,  there exist 
$i\in \{1,\ldots,\ell\}$ and an infinite subset $N_2$ of $N_1$ such $[s_n,s'_n]\subseteq F_i$ for each $n\in N_2$. 

Consider the set $H=\{t(v-u): \,t\in \R,\,\,u,v\in F_i\}$. By Lemma \ref{lem:SpanF} and since the space $\wt{X}$ 
is a real finite dimensional the set $H$ is closed. By Lemma \ref{lem:wh[P,A]}\eqref{item:HSwh[P,A]} the set $\wh{F_i,F_i}$ is nothing but the  intersection of $H$ and the unit sphere $S_{\wt{X}}$. Hence $\wh{F_i,F_i}$ is closed. Using \eqref{eq:a_np_ns_ns'_n} with $n\in N_2$ we conclude that $(a-p)/|a-p|$ is a limit point of elements from 
$\wh{F_i,F_i}$. Thus $(a-p)/|a-p|\in \wh{F_i,F_i}$. Hence $(a-p)/|a-p|=t(u-v)$ where $u,v\in F_i, u\neq v$ and 
$t\in \R$. Either $t=-1/|u-v|$ or $t=1/|u-v|$ and hence $(a-p)/|a-p|=\pm(v-u)/|v-u|$. 
Because $F_i$ is convex the line segment $[u,v]$ is contained in $F_i\subseteq S_{\wt{X}}$. 
From the preceding sentences  it follows that $[p,a]$ is parallel to $[u,v]$ which is a nondegenerate 
line segment contained in the unit sphere. This is a contradiction to the initial assumption. 
Hence $d(\wh{P,A},\wh{S_{\wt{X}}})>0$. 
\end{proof}

\begin{lem}\label{lem:NoParallelSites}
Suppose that $\min\{d(P_k,P_j): k,j\in K, k\neq j\}>0$ and that no two points from different sites form a segment which is parallel to a non-degenerate segment 
contained in the unit sphere of the normed space $\wt{X}$. Then there exists $r>0$ such that  
$d(\wh{P_k,A_k},\wh{S_{\wt{X}}})\geq r$ for all $k\in K$ where $A_k=\cup_{j\neq k}P_j$.
\end{lem}
\begin{proof}
By Lemma \ref{lem:NoParallelPA} we know that 
$r_{k,j}:=d(\wh{P_k,P_j},\wh{S_{\wt{X}}})>0$ for all $k,j\in K$, $j\neq k$. 
(If $\wh{S_{\wt{X}}}=\emptyset$, then any positive number, say 1, satisfies the assertion.) 
Hence $\min\{r_{k,j}:\, k,j\in K,\,\, j\neq k\}>0$. Now fix $k\in K$ and let $\phi\in\wh{P_k,A_k}$ and 
$\theta\in \wh{S_{\wt{X}}}$. Then $\phi=(p_j-p_k)/|p_j-p_k|$ for some $p_k\in P_k$ and $p_j\in P_j$, or 
$\phi=(q_k-q_j)/|q_k-q_j|$ for some $q_k\in P_k$ and $q_j\in P_j$. Thus  $|\phi-\theta|\geq r_{k,j}$. Hence $d(\wh{P_k,A_k},\wh{S_{\wt{X}}})\geq\min\{r_{k,j}: j\neq k, j\in K\}$. Therefore  every  positive $r$ satisfying $r\leq \min\{r_{k,j}: j,k\in K, j\neq k\}$  satisfies the desired conclusion. 
\end{proof}

\begin{lem}\label{lem:lambda_uni}
Let $\epsilon>0$ and $r>0$ be given. Let 
\begin{multline}\label{eq:Psi}
\Psi_{\epsilon,r}=\{(P,A): P\,\,\textnormal{and}\,\,A\,\textnormal{are nonempty  compact subsets of}\, X\,\,\textnormal{satisfying}\,\\ 
\, d(P,A)/2\geq \epsilon \,\,\textnormal{and}\,\, d(\wh{P,A},\wh{S_{\wt{X}}})\geq r\}. 
\end{multline}
 Then there exists a universal constant $\lambda\in (0,\epsilon)$ such that \eqref{eq:property} holds with this $\lambda$ for all $(P,A)\in \Psi_{\epsilon,r}$. 

\end{lem}

\begin{proof}
Note that the $\lambda$ described in Lemma \ref{lem:LambdaNUC} depends on $P$ and $A$, while here it is independent of them. To avoid ambiguity, we denote the $\lambda$ from Lemma \ref{lem:LambdaNUC} by $\lambda_{P,A}$. The assertion is 
obvious (void) if $\Psi_{\epsilon,r}=\emptyset$, so assume this set is nonempty. 
Given any $(P,A)\in \Psi_{\epsilon,r}$ it follows by 
Lemma \ref{lem:SegmentInequality} that  \eqref{eq:InSegment} holds since $d(\wh{P,A},\wh{S_{\wt{X}}})\geq r$. 
Hence Lemma \ref{lem:LambdaNUC} does imply the existence of $\lambda_{P,A}>0$ such that \eqref{eq:property} holds. 
Recall also that $\lambda_{P,A}<\epsilon$. Define 
\begin{equation}\label{eq:lambda_universal}
\lambda=\inf\{\lambda_{P,A}: (P,A)\in \Psi_{\epsilon,r}\}.
\end{equation}
For finishing the proof it suffices to prove that $\lambda>0$.
Suppose by way of contradiction that $\lambda=0$.  Then for any $n\in \N$ there exist compact subsets $P^{(n)}$  and $A^{(n)}$ of $X$ satisfying $d(P^{(n)},A^{(n)})/2\geq\epsilon$ and $d(\wh{P^{(n)},A^{(n)}},\wh{S_{\wt{X}}})\geq r$ 
 such that the number $\lambda_n:=\lambda_{P^{(n)},A^{(n)}}>0$ corresponding to them from Lemma \ref{lem:LambdaNUC} satisfies $\lambda_n<1/n$.  Note that by \eqref{eq:S_lambda_til} and \eqref{eq:lambda_epsilon} we have $\lambda_n=\wt{\lambda}_n<\epsilon/2$ for large enough $n$. This, \eqref{eq:S_lambda_til},  and the compactness of $A^{(n)}$  imply that there exist $p_n\in P^{(n)}$, $y_n\in X$, $x_n\in [p_n,y_n]$, and $a_n\in A^{(n)}$ such that $d(x_n,y_n)=\epsilon/2$, $d(y_n,p_n)\leq d(y_n,A^{(n)})=d(y_n,a_n)$ and 
\begin{equation}
\lambda_n\leq d(x_n,a_n)-d(x_n,p_n)<\lambda_n+1/n<2/n.
\end{equation}
Because $d(\wh{P^{(n)},A^{(n)}},\wh{S_{\wt{X}}})\geq r$ we also have 
\begin{equation}\label{eq:aps1s2}
|((a_n-p_n)/|a_n-p_n|)-((s_2-s_1)/|s_2-s_1|)|\geq r
\end{equation}
for every nondegenerate segment $[s_1,s_2]\subseteq S_{\wt{X}}$. By passing to convergent subsequences we obtain the existence of points $p,x,y$, and $a$ contained in $X$ and satisfying the relations $x\in [p,y]$, $d(x,y)=\epsilon/2$, $d(y,p)\leq d(y,a)$, $d(p,a)/2\geq\epsilon$ and $d(x,a)-d(x,p)=0$. If $x=p$, then $p=x=a\in A$, contradicting the fact that $d(P,A)>0$. 
Hence $d(x,a)=d(x,p)>0$. From Lemma \ref{lem:ParallelUnitSphere}  it follows that 
$[p,a]$ is parallel to a nondegenerate segment $[s_1,s_2]$ that is contained in the unit sphere $S_{\wt{X}}$. 
But this is impossible since   $|((a-p)/|a-p|)-((s_2-s_1)/|s_2-s_1|)|\geq r$ because of \eqref{eq:aps1s2},  a contradiction. 
\end{proof}

We are now ready to prove Theorem \ref{thm:stabilityNUC}. 

\begin{proof}[{\bf Proof of Theorem \ref{thm:stabilityNUC}:}] 
It is sufficient to assume that $\epsilon\in(0,\eta/6)$, otherwise one can take 
for the given $\epsilon$ the $\Delta$ associated with, say, $\eta/10$. 
We will show that all the conditions needed in 
Proposition ~\ref{prop:stabilityVor} are satisfied. 

Since $X$ is bounded we obtain  \eqref{eq:Mk} with, say, $M=M'=\textnormal{diam}{X}$. 
Let $k\in K$ be given. Since no two points of $P_k$ and $A_k=\cup_{j\neq k}P_j$ 
form a segment which is parallel to a non-degenerate segment contained in the unit sphere, 
Lemma \ref{lem:NoParallelSites} implies the existence of $r>0$ such that $d(\wh{P_k,A_k},\wh{S_{\R^m}})\geq 2r$ for each $k\in K$. By \eqref{eq:eta} it follows that $d(P_k,A_k)\geq \eta>2\epsilon$ for each $k\in K$. Therefore  $(P_k,A_k)\in \Psi_{\epsilon,r}$ for each $k\in K$ (see \eqref{eq:Psi}; note that we take here $\Psi_{\epsilon,r}$ and not $\Psi_{\epsilon,2r}$).  

Let $\lambda$ be taken from Lemma \ref{lem:lambda_uni} and let $\lambda'=\lambda$. Let $\Delta$ 
be any positive number satisfying 
\begin{equation*}
\Delta<\min\left\{\frac{\lambda}{8(1+(M/\epsilon))}, \frac{r\eta}{8}\right\}.
\end{equation*}
In particular $\Delta<\lambda<\epsilon$. Let $(P'_k)_{k\in K}$ be any tuple of nonempty compact subsets of $X$ satisfying $D(P_k,P'_k)<\Delta$ for any $k\in K$. 
For each $k\in K$ let $A'_k=\cup_{j\neq k}P_j$. Then $D(A_k,A'_k)\leq \min\{ D(P_j,P'_j): j\in K, j\neq k\}< \Delta<\epsilon$  by Lemma \ref{lem:HausdorffUnion}. This and Lemma \ref{lem:whPAP'A'}\eqref{item:d(P,A)eps1eps2} imply that 
\begin{equation}\label{eq:P'A'eps}
d(P'_k,A'_k)\geq d(P_k,A_k)-2\epsilon\geq \eta-2\epsilon> 2\epsilon. 
\end{equation}
In addition, Lemma \ref{lem:whPAP'A'}\eqref{item:wh[P,A]} implies that $D(\wh{P_k,A_k},\wh{P'_k,A'_k})<4\Delta/\eta <r$. 

Fix $k\in K$ and let $s\in \wh{S_{\R^m}}$ and $q'\in \wh{P'_k,A'_k}$  be given. 
From the inequality just proved, namely $D(\wh{P_k,A_k},\wh{P'_k,A'_k})<r$, there exists  $q\in \wh{P_k,A_k}$ such that $d(q,q')< r$. In addition, $d(q,s)\geq 2r$ because $d(\wh{P_k,A_k},\wh{S_{\R^m}})\geq 2r$. Hence 
\begin{equation*}
2r\leq d(q,s)\leq d(q,q')+d(q',s)<r+d(q',s). 
\end{equation*}
As a result, $d(\wh{P'_k,A'_k},\wh{S_{\R^m}})\geq r$ since $q'$ and $s$ were arbitrary. 
Consequently (see \eqref{eq:Psi}), this and \eqref{eq:P'A'eps} imply that $(P'_k,A'_k)\in \Psi_{\epsilon,r}$ for each $k\in K$. Hence Lemma \ref{lem:lambda_uni} implies that \eqref{eq:property} holds with $P=P_k$ and $A=A_k$, and also with $P=P'_k$ and $A=A'_k$ for each $k\in K$. Therefore  \eqref{eq:property_k} and \eqref{eq:property'_k} hold with $\lambda'=\lambda$.  Thus, by denoting $\epsilon_4=\Delta$ we see that all the assumptions in Proposition ~\ref{prop:stabilityVor} are satisfied, and hence $D(R_k,R'_k))<\epsilon$ holds for each $k\in K$.
\end{proof}

The proof of Theorem \ref{thm:TopologicalProperties} is based on Lemma  ~\ref{lem:SegmentInequality}.  

\begin{proof}[{\bf Proof of Theorem \ref{thm:TopologicalProperties}}]
First \eqref{eq:closure} will be proved. 
Let $f:X\to \R$ be defined by  $f(x)=d(x,P)-d(x,A)$ for all $x\in X$. This  
is a continuous function (with respect to the metric topology). Since $\dom(P,A)=\{x\in X: f(x)\leq 0\}$ 
it follows that $\dom(P,A)$ is closed. This and the inclusion $\{x\in X: d(x,P)<d(x,A)\}\subseteq \dom(P,A)$ 
imply the inclusion $\overline{\{x\in X: d(x,P)<d(x,A)\}}\subseteq \dom(P,A)$.  
 For the reverse inclusion, let $z\in \dom(P,A)$. If $d(z,P)<d(z,A)$, then obviously the point $z$ is in the set $\overline{\{x\in X: d(x,P)<d(x,A)\}}$. Now suppose that $d(z,P)=d(z,A)$.
By assumption there is $p\in P$ such that $d(z,P)=d(z,p)$, and since $p\in P$, it  follows that  $p\notin A$. This 
and the fact that $d(p,A)$ is attained imply that $d(p,A)>0$. Hence if $z=p$, then $z\in \overline{\{x\in X: d(x,P)<d(x,A)\}}$.  If $z\neq p$, then $[p,z)\neq \emptyset$, and Lemma  ~\ref{lem:SegmentInequality} 
implies that every $x\in [p,z)$ (arbitrary close to $z$) satisfies the inequality $d(x,P)\leq d(x,p)<d(x,A)$. 
Thus again $z\in \overline{\{x\in X: d(x,P)<d(x,A)\}}$.

Now \eqref{eq:boundary} will be proved. Because of the continuity of the function $f$ defined above 
one obtains that  all the points in $\{x\in X: d(x,P)<d(x,A)\}$ and $\{x\in X: d(x,A)<d(x,P)\}$ 
are interior points. It follows that $\partial(\dom(P,A))\subseteq \{x\in X: d(x,P)=d(x,A)\}$ 
without any assumption on the sites or on $X$. For the reverse inclusion, let 
 $z$ satisfy $d(z,P)=d(z,A)$. Then $z\in \dom(P,A)$. Let $a\in A$ satisfy $d(z,a)=d(z,A)$. 
 Then $a\neq z$ since otherwise the equality $d(z,P)=d(z,A)$ and the fact that $d(z,P)$ is 
 attained would imply that $z\in P\cap A$, a contradiction. Hence $[a,z)\neq\emptyset$. The 
 inclusion  $[a,z)\subseteq \{x\in X: d(x,A)<d(x,P)\}$ holds  
because of Lemma  ~\ref{lem:SegmentInequality} (with $P$ instead of $A$ 
and $a$ instead of $p$) and it proves that arbitrary close to $z$ 
there are points outside $\dom(P,A)$. Thus 	\eqref{eq:boundary} holds. 

Finally, \eqref{eq:boundary}, the equality 
\begin{multline*}
\partial(\dom(P,A))\bigcup \Int(\dom(P,A))=\dom(P,A)\\= \{x\in X: d(x,P)=d(x,A)\}\bigcup\{x\in X: d(x,P)<d(x,A)\},    
\end{multline*}
 and the fact that the terms in both unions are disjoint, all imply \eqref{eq:interior}. 
\end{proof}

The proof is of Corollary ~\ref{cor:StabilityBisectors} is based on Theorem ~\ref{thm:stabilityNUC} 
and Theorem  ~\ref{thm:TopologicalProperties}. 
It is also based on the 
fact that the union of the Voronoi cells is the whole world $X$, i.e., 
no neutral Voronoi region exists. The non-existence of a neutral region follows 
immediately from the fact that finitely many sites are considered (but there is a wide 
class of cases where it does not exist even when infinitely many sites 
are considered).

\begin{proof}[{\bf Proof of Corollary ~\ref{cor:StabilityBisectors}}]

Let $k\in K$ be given. 
Suppose for a contradiction that $D(B_k,B'_k)>\epsilon$. This means that 
either there exists $x\in B_k$ satisfying $d(x,B'_k)>\epsilon$ or 
 there exists $x'\in B'_k$ satisfying $d(x',B_k)>\epsilon$. 
Assume the first case. The second one can be proved in the same way. 
Let $\tau>0$ be small enough such that  $d(x,B'_k)>\tau+\epsilon$. 

Consider the function $f:X\to\R$ defined by $f(z)=d(z,P'_k)-d(z,A'_k)$. 
Since $x\notin B'_k$ either $f(x)>0$ or $f(x)<0$. 
It must be that the sign of $f(z)$  is the same as of $f(x)$ for all 
$z\in X$ in the open ball $B(x,\tau+\epsilon)$. Indeed, if this not true, 
then for some $z\in X$ in the ball either $f(z)=0$ or 
both $f(z)$ and $f(x)$ have opposite signs. But in both cases the intermediate value theorem implies that the 
continuous function $g:[0,1]\to \R$ defined by $g(t)=f(tx+(1-t)z)$ vanishes at some 
$t\in [0,1]$, that is, $f(y_t)=0$ for $y_t=tx+(1-t)z\in X$. It follows that $y_t\in B'_k\cap B(x,\tau+\epsilon)$. 
In particular $d(x,B'_k)\leq d(x,y_t)<\tau+\epsilon$, a contradiction to the assumption 
$d(x,B'_k)>\tau+\epsilon$.

Assume now that $f(x)>0$. By the definition of $f$ this implies that $d(x,A'_k)<d(x,P'_k)$. 
Since $A'_k=\cup_{j\neq k}P'_j$ it follows that there exists $j\neq k$ such that $d(x,P'_j)<d(x,P'_k)$. 
Since $x\in R_k$ and since $D(R_k,R'_k)<\epsilon$ 
(by Theorem ~\ref{thm:stabilityNUC}) there exists $x'\in R'_k$ such that $d(x,x')<\epsilon$. 
By definition $d(x',P'_k)\leq d(x',A'_k)$, i.e., $f(x')\leq 0$. 
But $x'\in B(x,\tau+\epsilon)$ and the sign of $f(x')$ is not the same as $f(x)$, 
a contradiction to the previous paragraph. 

It remains to consider the case  $f(x)<0$. Since by assumption $x\in B_k$ and 
since $B_k=\partial(R_k)$ by Theorem  ~\ref{thm:TopologicalProperties}, 
there exists $y\in X\backslash R_k$ such that $d(y,x)<\tau$. Since $y\notin R_k$ 
and since $\bigcup_{j\in K}R_j=X$ as explained before the formulation of the theorem, 
there exists $j\neq k$ such that $y\in R_j$. Since $D(R_j,R'_j)<\epsilon$ 
(by Theorem ~\ref{thm:stabilityNUC}) there exists $y'\in R'_j$ such that $d(y,y')<\epsilon$. 
By definition $d(y',P'_j)\leq d(y',A'_j)$. In particular, 
\begin{equation}\label{eq:y'jk}
d(y',P'_j)\leq d(y',P'_k). 
\end{equation}
But $d(y',x)\leq d(y',y)+d(y,x)<\epsilon+\tau$. Hence $y'\in B(x,\tau+\epsilon)$ 
and by what proved earlier the sign of $f(y')$ must be the same as of $f(x)$. Thus $f(y')<0$, 
that is, $d(y',P'_k)<d(y',A'_k)\leq d(y',P'_j)$, a contradiction to \eqref{eq:y'jk}. 

The above mentioned contradictions show that the assumption 
$D(B_k,B'_k)>\epsilon$ cannot hold. Does $D(B_k,B'_k)\leq\epsilon$ for each $k\in K$, as claimed. 
\end{proof} 

\section*{Acknowledgments}
I thank Eytan Paldi for a suggestion which helped to shorten the proof of Lemma ~\ref{lem:SpanF}. \\

{\noindent \bf Additional acknowledgments: To be added.}

\bibliographystyle{amsplain}
\bibliography{biblio}

\end{document}